\documentclass[12pt,a4paper]{article}
\usepackage[a4paper,margin=2.55cm]{geometry}
\usepackage[T1]{fontenc}
\usepackage[utf8]{inputenc}
\usepackage{lmodern}
\usepackage{microtype}
\usepackage{amsmath,amssymb,amsthm}
\usepackage{booktabs}
\usepackage{array}
\usepackage{url}
\usepackage{xurl}
\usepackage{xcolor}
\usepackage{tikz}
\usetikzlibrary{arrows.meta,calc,patterns,positioning}
\usepackage{algorithm}
\usepackage{algpseudocode}
\usepackage{hyperref}
\usepackage{orcidlink}
\hypersetup{colorlinks=true,linkcolor=blue!60!black,citecolor=blue!60!black,urlcolor=blue!60!black}

\newcommand{\Rtwo}{R_2}
\newcommand{\Rtwoavg}{\overline{R}_2}
\newcommand{\dd}{\,\mathrm d}
\newcommand{\RR}{\mathbb R}
\newcommand{\DeltaTwo}{\Delta_2}
\newcommand{\DeltaN}{\Delta_{N-1}}
\newcommand{\boxop}{\operatorname{box}}

\newcommand{\HV}{\operatorname{HV}}

\newtheorem{theorem}{Theorem}
\newtheorem{proposition}{Proposition}
\newtheorem{lemma}{Lemma}
\newtheorem{corollary}{Corollary}
\newtheorem{remark}{Remark}

\title{Computing the Integral \texorpdfstring{$R_2$}{R2} Indicator by \\
Perspective Mapping and Box Decomposition}
\author{Michael T. M. Emmerich\,\orcidlink{0000-0002-7342-2090}\\
\small Faculty of Information Technology, University of Jyväskylä, Finland}
\date{}

\begin{document}
\maketitle

\begin{abstract}
The continuous integral $R_2$ indicator is a Pareto-compliant refinement of the classical finite-weight-vector $R_2$ indicator, used in performance assessment, bounded archiving for a-posteriori multi-objective optimization, and skyline selection in databases. This work introduces a bidirectional perspective mapping between continuous integral $R_2$ computation and integration over unions of anchored axis-aligned boxes. After translating the ideal point of a minimization problem to the origin, approximation points become strictly positive loss vectors, and the subgraph of the lower weighted Tchebycheff envelope over the weight simplex maps to the complement of an anchored-box union in reciprocal objective space. The Jacobian gives an absolute $R_2$ formula as a weighted complement volume with density $(x_1+\cdots+x_N)^{-(N+1)}$, while differences of $R_2$ values become finite weighted hypervolume differences. Hence, hypervolume algorithms that emit box decompositions can be reused by replacing ordinary box volumes with closed-form weighted box integrals. For $N$ objectives, this gives an output-sensitive overhead $O(2^N M)$ for an $M$-box decomposition, or $O(M)$ for fixed $N$. Using existing box-decomposition approaches, the integral $R_2$ can be computed in $O(n \log n)$ for $N=2,3$, in $O(n^2)$ for $N=4$, and in $O\left(n^{\lfloor (N-1)/2\rfloor+1}\right)$ for $N\geq4$, with $n$ denoting the size of the approximation set. On the lower-bound side, exact value computation has an $\Omega(n\log n)$ lower bound in the algebraic decision-tree model already in two objectives, this bound lifts to every fixed $N\geq2$, and exact computation is $\#P$-hard when $N$ is part of the input. Together, the proposed perspective mapping provides a powerful tool for transferring algorithmic and structural results between anchored-box union and hypervolume theory and integral $R_2$ computation.
\end{abstract}
\section{Introduction}
The integral $R_2$ indicator evaluates a Pareto front approximation set by integrating the best weighted Tchebycheff scalarizing value over a continuous weight domain \cite{schaepermeier2024r2,schaepermeier2025ecj,jaszkiewicz2025exact}. It is a continuous refinement of the classical $R_2$ indicator \cite{hansen1998evaluating}, which uses a finite set of weight vectors. Unlike the classical $R_2$ indicator, the integral $R_2$ is Pareto compliant, making it suitable for performance assessment and bounded archiving. In contrast to the hypervolume indicator, which requires a dystopian reference point, the $R_2$ indicators use an ideal reference point, often easier to specify for a given problem. Each weight vector defines a scalar preference model; the lower envelope of the corresponding Tchebycheff shadows \cite{emmerich2026preferenceshaped} gives the best scalarizing value achieved by the approximation set. Exact computation is desirable because finite weight sampling may miss small regions of the weight simplex and thus fail to capture fine variations of this lower envelope.

This report develops a geometric route to exact computation. In short, we will see that the computation of the integral $R_2$ indicator can be accomplished through computations of closed-form integrals on axis-aligned boxes. The central instrument for achieving this is a perspective mapping that transforms the coordinates of relevant points - kinks of the Tchebycheff shadows in the Tchebycheff weight-loss representation - in the integration region into corner points of axis-aligned cuboids. The Jacobian determinant \cite{Jacobi1896Functionaldeterminanten} of this mapping is variable, but all resulting integrals and computations are closed-form and exact. The same mapping can be used in the inverse direction, which is useful for transferring structural results, such as complexity bounds, from the well-developed computational theory of the hypervolume indicator \cite{guerreiro2021hypervolume}. 

Let us discuss the proposed mapping in more details next:
After translating the ideal point of a minimization problem to the origin, approximation points become strictly positive loss vectors. The key observation is that the perspective transformation
\(x_i=\frac{w_i}{t}\)
uses the coordinates of a weight vector $w=(w_1,\ldots,w_N)\in\DeltaN$ and the positive height variable $t>0$ in the subgraph of the Tchebycheff lower envelope. In these coordinates, the subgraph is transformed into the complement of an anchored-box union in reciprocal objective space. The Jacobian of this map is not constant. Hence, the integral $R_2$ is not an ordinary hypervolume value, but a \emph{weighted} hypervolume-type integral. Nevertheless, box decompositions from hypervolume computation can be reused: once the reciprocally dominated region is decomposed into disjoint axis-aligned boxes, each ordinary box-volume term is replaced by a closed-form weighted box integral.

By the perspective mapping introduced in this work and its related transformations, we will derive exact computation algorithms, new algorithmic upper bounds, and value-computation lower bounds for the integral $R_2$ indicator for $N\geq2$.  The previously known $O(n\log n)$ algorithm covers the two-dimensional case \cite{schaepermeier2025ecj}; for $N>2$, the present approach yields improved worst-case upper bounds by reusing hypervolume-style box decompositions.  In three objectives, the dimension-sweep algorithm of Fonseca, Paquete, and L{\'o}pez-Ib{\'a}nez~\cite{fonseca2006dimension} provide an $O(n\log n)$ box-emitter with $M=O(n)$ boxes, and thus an $O(n \log n)$ time complexity algorithm for the Integral R2 indicator in three dimensions.  Likewise, we obtain $O(n^2)$ time complexity for $N=4$ using \cite{Guerreiro12}, and $O\!\left(n^{\lfloor (N-1)/2\rfloor+1}\right)$ for fixed $N\geq4$ using \cite{Lacour17}.  Conversely, exact value computation has an $\Omega(n\log n)$ lower bound already in two objectives by reduction from uniform gap \cite{beume2009complexity,benor1983lower,preparata1985computational}; this lower bound lifts to every fixed $N\geq2$, and a perspective-weighted version of the Bringmann--Friedrich anchored-box reduction \cite{bringmann2010volume} gives $\#P$-hardness when $N$ is part of the input.

The report is organized as follows. Section~\ref{sec:setting} fixes the notation. Section~\ref{sec:warmup2d} gives a two-dimensional warm-up with a visualization of the mapping and a hand calculation. Section~\ref{sec:absolute3d} discusses the three-dimensional absolute and improvement formulas and the $O(n\log n)$ consequence. Section~\ref{sec:Ndim} extends the construction to $N$ objectives and provides pseudocode, correctness, output-sensitive upper bounds, and value-computation lower bounds. Section~\ref{sec:relatedwork} relates this work to the published results in the literature, while the appendix gives detailed lower-bound proofs and verification scripts.

\section{Setting and notation}\label{sec:setting}
We consider minimization and a finite approximation set to a Pareto front of a multiobjective optimization or skyline computation problem with $N$ objective functions denoted with $P$.  The ideal point is translated to the origin. For the standard background in multiobjective optimization we refer to Miettinen \cite{miettinen1999nonlinear}. The standing assumption is therefore
\[
  P=\{p^1,\ldots,p^n\}\subset\RR^N_{>0}.
\]
and points in $p$ are mutually non-dominated in the Pareto dominance relation.
Equivalently, in the original objective space, the ideal point strictly dominates every approximation point.  The coordinate $p_i$ is the positive loss of point $p$ in objective $i$.
For $N$ objectives the weight domain is
\[
  \DeltaN=\{w\in\RR^N_{\geq0}: w_1+\cdots+w_N=1\}.
\]
For a point $p\in\RR^N_{>0}$ define its weighted Tchebycheff shadow \cite{emmerich2026preferenceshaped} by
\[
  g_p(w)=\max_{i=1,\ldots,N} w_i p_i.
\]
For a finite set $P$, the lower Tchebycheff envelope is the area (or subgraph) under the union of the Tchebycheff shadows of all points of the approximation set $P$
\[
  \tau_P(w)=\min_{p\in P} g_p(w),
\]
and the unnormalised integral $R_2$ value is
\[
  \Rtwo(P)=\int_{\DeltaN}\tau_P(w)\dd w .
\]
The simplex-normalized average is obtained by multiplying this unnormalized value by $(N-1)!$, since $\operatorname{vol}_{N-1}(\Delta_N)=1/(N-1)!$. Importantly, this is \textit{not} the classical $R_2$ indicator based on a finite set of weight vectors, but its Pareto-compliant, continuous integral-based counterpart, which was proposed more recently by Schaepermeier and Kerschke \cite{schaepermeier2024r2} and Jaszkiewicz and Zielniewicz \cite{jaszkiewicz2025exact}. 

For the geometrical reasoning, the Tchebycheff loss and the weight space view developed in \cite{emmerich2026preferenceshaped} and \cite{emmerich2026threeobjectiveintegralr2subset} are useful. Here, the Tchebycheff shadows and the subgraph of their lower envelope are viewed in a space that maps weight vectors $w \in \Delta_{N-1}$ to \textit{Tchebycheff loss values}, that is values of the lower envelope $t =\tau_P(w)$. In 2-D and 3-D, this leads to simple geometric visualizations of the integration regions, and the reader is referred to \cite{emmerich2026threeobjectiveintegralr2subset} for examples.

The integral $R_2$ indicator can also be viewed as a relative improvement indicator given a reference point or set \cite{emmerich2026threeobjectiveintegralr2subset}. A point $a\in\RR^N_{>0}$ is called a dominated anchor for $P$ if
\[
  a_i\geq p_i \qquad\text{for all }p\in P\text{ and all }i.
\]
The anchor-normalised improvement is
\(
  I_a(P)=\Rtwo(\{a\})-\Rtwo(P).
\)
The improvement is convenient for subset selection using greedy approximation \cite{emmerich2026threeobjectiveintegralr2subset}.  The absolute value $\Rtwo(P)$ is recovered exactly from the same computation by adding the single-anchor value $\Rtwo(\{a\})$, where $a$ can be viewed as the equivalent of a dystopian reference point of the hypervolume indicator.

\section{Warm-up in two objectives}\label{sec:warmup2d}
We begin with one concrete two-objective example and compute the integral $R_2$ from first principles.  Let
\[
  p=(2,1).
\]
The weight simplex is the interval $w\in[0,1]$, with weights $(w,1-w)$.  The weighted Tchebycheff value of the point is
\[
  g_p(w)=\max\{2w,1-w\}.
\]
The kink is found by equating the two affine pieces,
\[
  2w=1-w,
  \qquad
  w_0=\frac13,
  \qquad
  g_p(w_0)=\frac23.
\]
Thus the kink in the $(w,t)$-diagram is
\[
  q_0=(w_0,t_0)=\left(\frac13,\frac23\right).
\]
Figure~\ref{fig:warmup-shadow-reciprocal} shows the whole one-point shadow before we look at local area elements.

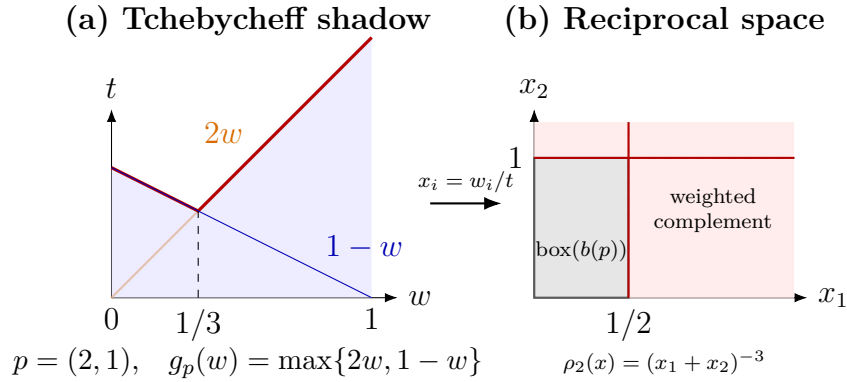
\begin{figure}[!htb]
\centering
\begin{tikzpicture}[scale=0.86,>=Latex]
  \begin{scope}[shift={(0,0)}]
    
    \draw[->] (0,0) -- (4.4,0) node[right] {$w$};
    \draw[->] (0,0) -- (0,2.9) node[above] {$t$};
    
    \draw[orange!85!black,thick] (0,0) -- (4,4.0) node[pos=0.55,above left] {$2w$};
    \fill[blue!10,opacity=0.7] (0,0) -- (0,2.0) -- (1.333,1.333) -- (4,4.0) -- (4,0) -- cycle;
    \draw[very thick,red!70!black] (0,2.0) -- (1.333,1.333) -- (4,4.0);
    \draw[dashed] (1.333,0) node[below] {$1/3$} -- (1.333,1.333);
    \node[below] at (0,0) {$0$};
    \node[below] at (4,0) {$1$};
    \node[align=center,font=\small] at (2.1,-1.0) {$p=(2,1)$,\quad $g_p(w)=\max\{2w,1-w\}$};
    \node[font=\bfseries] at (2.1,4.25) {(a) Tchebycheff shadow};
    \draw[blue!70!black] (0,2.0) -- (4,0) node[pos=0.78,above right] {$1-w$};
  \end{scope}

  \draw[->,thick] (4.9,1.45) -- node[above,align=center,font=\scriptsize] {$x_i=w_i/t$} (6.0,1.45);

  \begin{scope}[shift={(6.5,0)}]
    \node[font=\bfseries] at (2.0,4.25) {(b) Reciprocal space};
    \draw[->] (0,0) -- (4.2,0) node[right] {$x_1$};
    \draw[->] (0,0) -- (0,2.9) node[above] {$x_2$};
    \fill[gray!20] (0,0) rectangle (1.45,2.15);
    \draw[gray!70!black,thick] (0,0) rectangle (1.45,2.15);
    \fill[red!10,opacity=0.75] (1.45,0) rectangle (4.0,2.7);
    \fill[red!10,opacity=0.75] (0,2.15) rectangle (1.45,2.7);
    \draw[dashed] (1.45,0) node[below] {$1/2$} -- (1.45,2.7);
    \draw[dashed] (0,2.15) node[left] {$1$} -- (4.0,2.15);
    \draw[red!70!black,thick] (1.45,0) -- (1.45,2.7);
    \draw[red!70!black,thick] (0,2.15) -- (4.0,2.15);
    \node[align=center,font=\scriptsize] at (2.75,1.35) {weighted\\complement};
    \node[align=center,font=\scriptsize] at (2.0,-1) {$\rho_2(x)=(x_1+x_2)^{-3}$};
    \node[align=center,font=\scriptsize] at (0.73,0.72) {$\boxop(b(p))$};
  \end{scope}
\end{tikzpicture}
\caption{View for $p=(2,1)$. The V-shaped Tchebycheff shadow maps to the complement of the reciprocal box with corner $b(p)=(1/2,1)$ and density $\rho_2(x)=(x_1+x_2)^{-3}$.}
\label{fig:warmup-shadow-reciprocal}
\end{figure}

The scalarization-side computation is now just the area under the V-shaped graph.  Splitting the integral at the kink gives
\[
\begin{aligned}
  R_2(\{p\})
  &=\int_0^{1/3}(1-w)\,\mathrm dw
    +\int_{1/3}^1 2w\,\mathrm dw  \\
  &=\left[w-\frac{w^2}{2}\right]_0^{1/3}
    +\left[w^2\right]_{1/3}^{1} \\
  &=\frac{5}{18}+\frac{8}{9}
   =\frac{7}{6}.
\end{aligned}
\]
The same number will be recovered in reciprocal coordinates.  The perspective map is
\[
  \Phi(w,t)=\left(\frac{w}{t},\frac{1-w}{t}\right)=(x_1,x_2).
\]
Its inverse is
\[
  w=\frac{x_1}{x_1+x_2},
  \qquad
  t=\frac{1}{x_1+x_2}.
\]
For the kink point $q_0=(1/3,2/3)$ this gives
\[
  \Phi(q_0)
  =\left(\frac{1/3}{2/3},\frac{2/3}{2/3}\right)
  =\left(\frac12,1\right)
  =b(p).
\]
The two inequalities that define the subgraph below the Tchebycheff shadow become reciprocal-space inequalities.  Indeed,
\[
  t\le 2w
  \quad\Longleftrightarrow\quad
  \frac{w}{t}\ge\frac12
  \quad\Longleftrightarrow\quad
  x_1\ge\frac12,
\]
and
\[
  t\le 1-w
  \quad\Longleftrightarrow\quad
  \frac{1-w}{t}\ge 1
  \quad\Longleftrightarrow\quad
  x_2\ge1.
\]
Hence the V-shaped subgraph maps to the complement of the anchored reciprocal box
\[
  [0,1/2]\times[0,1].
\]
Up to boundaries of zero area, this complement is the disjoint union
\[
  [1/2,\infty)\times[0,\infty)
  \;\cup\;
  [0,1/2]\times[1,\infty).
\]

It remains to account for how area changes under the map. We use Jacobi's theory (\cite{Jacobi1896Functionaldeterminanten}, p.53) to map differential area elements. The Jacobian matrix of $\Phi$ is
\[
  D\Phi(w,t)
  =
  \frac{\partial(x_1,x_2)}{\partial(w,t)}
  =
  \begin{pmatrix}
    \dfrac{1}{t} & -\dfrac{w}{t^2}\\[1.2ex]
    -\dfrac{1}{t} & -\dfrac{1-w}{t^2}
  \end{pmatrix}.
\]
Therefore
\[
\begin{aligned}
  \det D\Phi(w,t)
  &=
  \frac{1}{t}\left(-\frac{1-w}{t^2}\right)
  -
  \left(-\frac{w}{t^2}\right)\left(-\frac{1}{t}\right) \\
  &=-\frac{1-w}{t^3}-\frac{w}{t^3}
   =-\frac{1}{t^3}.
\end{aligned}
\]
Consequently, differential areas transform as
\[
  \mathrm dx_1\,\mathrm dx_2
  =\frac{1}{t^3}\,\mathrm dw\,\mathrm dt.
\]
Using $t=1/(x_1+x_2)$, the inverse area element is
\[
  \mathrm dw\,\mathrm dt
  =\frac{1}{(x_1+x_2)^3}\,\mathrm dx_1\,\mathrm dx_2.
\]
This is the density factor that appears in reciprocal space.

\begin{figure}[!htb]
\centering
\begin{minipage}{\textwidth}
\centering
\begin{minipage}[t]{0.49\textwidth}
\centering
\begin{tikzpicture}[scale=2,>=Latex,font=\scriptsize]
  \node[font=\bfseries] at (0.55,1.15) {$(w,t)$};
  \node[font=\bfseries] at (2.35,1.15) {$(x_1,x_2)$};

  \draw[->] (0,0) -- (1.05,0) node[below] {$w$};
  \draw[->] (0,0) -- (0,1.05) node[left] {$t$};
  \coordinate (A) at (0.33,0.67);
  \coordinate (B) at (0.49,0.67);
  \coordinate (C) at (0.49,0.83);
  \coordinate (D) at (0.33,0.83);
  \fill[blue!8] (A)--(B)--(C)--(D)--cycle;
  \draw[thick,blue!80!black] (A)--(B)--(C)--(D)--cycle;
  \fill (A) circle (0.010) node[below left] {$q_0$};
  \draw[->,thick,blue!80!black] (A)--(B) node[midway,below] {$dw$};
  \draw[->,thick,blue!80!black] (A)--(D) node[midway,left] {$dt$};
  \node[blue!80!black] at (0.53,0.38) {$dw\,dt$};

  \draw[->,very thick] (1.18,0.54)--(1.58,0.54) node[midway,above] {$\Phi$};

  \begin{scope}[shift={(1.75,0)}]
    \draw[->] (0,0) -- (1.05,0) node[below] {$x_1$};
    \draw[->] (0,0) -- (0,1.05) node[left] {$x_2$};
    \coordinate (P) at (0.36,0.74);
    \coordinate (Q) at (0.62,0.48);
    \coordinate (S) at (0.20,0.50);
    \coordinate (R) at ($(Q)+(S)-(P)$);
    \fill[orange!12] (P)--(Q)--(R)--(S)--cycle;
    \draw[thick,orange!85!black] (P)--(Q)--(R)--(S)--cycle;
    \fill (P) circle (0.010) node[above right] {$b(p)$};
    \draw[->,thick,orange!85!black] (P)--(Q) node[midway,above right] {$\vec v_w$};
    \draw[->,thick,orange!85!black] (P)--(S) node[midway,left] {$\vec v_t$};
    \node[orange!85!black] at (0.48,0.12) {$\tfrac{27}{8}dw\,dt$};
  \end{scope}
\end{tikzpicture}

\smallskip
{\small Forward mapping of a local area element.}
\end{minipage}\hfill
\begin{minipage}[t]{0.49\textwidth}
\centering
\begin{tikzpicture}[scale=2,>=Latex,font=\scriptsize]
  \node[font=\bfseries] at (0.55,1.15) {$(w,t)$};
  \node[font=\bfseries] at (2.35,1.15) {$(x_1,x_2)$};

  \draw[->] (0,0) -- (1.05,0) node[below] {$w$};
  \draw[->] (0,0) -- (0,1.05) node[left] {$t$};
    \coordinate (M) at (0.33,0.67);
    \coordinate (N) at (0.49,0.47);
    \coordinate (S2) at (0.22,0.49);
    \coordinate (O) at ($(N)+(S2)-(M)$);
    \fill[blue!8] (M)--(N)--(O)--(S2)--cycle;
    \draw[thick,blue!80!black] (M)--(N)--(O)--(S2)--cycle;
    \fill (M) circle (0.010) node[above right] {$q_0$};
    \draw[->,thick,blue!80!black] (M)--(N) node[midway,above right] {$\vec u_1$};
    \draw[->,thick,blue!80!black] (M)--(S2) node[midway,left] {$\vec u_2$};
    \node[blue!80!black] at (0.50,0.14) {$\tfrac{8}{27}dx_1\,dx_2$};

  \draw[<-,very thick] (1.18,0.54)--(1.58,0.54) node[midway,above] {$\Phi^{-1}$};

  \begin{scope}[shift={(1.75,0)}]
    \draw[->] (0,0) -- (1.05,0) node[below] {$x_1$};
    \draw[->] (0,0) -- (0,1.05) node[left] {$x_2$};
    \coordinate (E) at (0.36,0.74);
    \coordinate (F) at (0.54,0.74);
    \coordinate (G) at (0.54,0.56);
    \coordinate (H) at (0.36,0.56);
    \fill[orange!12] (E)--(F)--(G)--(H)--cycle;
    \draw[thick,orange!85!black] (E)--(F)--(G)--(H)--cycle;
    \fill (E) circle (0.010) node[above left] {$b(p)$};
    \draw[->,thick,orange!85!black] (E)--(F) node[midway,above] {$dx_1$};
    \draw[->,thick,orange!85!black] (E)--(H) node[midway,left] {$dx_2$};
    \node[orange!85!black] at (0.49,0.36) {$dx_1\,dx_2$};
  \end{scope}
\end{tikzpicture}

\smallskip
{\small Inverse mapping of a local area element.}
\end{minipage}
\end{minipage}

\medskip
\begin{minipage}{\textwidth}
\centering
\begin{tikzpicture}[scale=0.8,transform shape,>=Latex,font=\scriptsize]
\begin{scope}[shift={(0,0)}]
\node[font=\bfseries] at (2.700,6.480) {$ (w,t)$-space};
\draw[->] (0,0) -- (5.750,0) node[below] {$w$};
\draw[->] (0,0) -- (0,6.150) node[left] {$t$};
\draw (2.6544,0.05) -- (2.6544,-0.05) node[below] {$\tfrac13$};
\draw (3.9498,0.05) -- (3.9498,-0.05) node[below] {$0.4$};
\draw (0.05,2.5346) -- (-0.05,2.5346) node[left] {$\tfrac23$};
\draw (0.05,4.5826) -- (-0.05,4.5826) node[left] {$0.8$};
\draw[densely dashed,blue!55] (2.6544,0) -- (2.6544,2.5346) -- (0,2.5346);
\draw[gray,thick,dashed] (0.0000,4.6328) -- (5.4000,0.3642);
\node[gray!80!black,font=\bfseries,anchor=east] at (4.9300,0.7) {$t=1-w$};
\draw[gray,thick,dotted] (1.0512,0.0000) -- (4.7199,5.8000);
\node[gray!80!black,font=\bfseries,anchor=west] at (4.0055,4.7200) {$t=2w$};
\fill[cyan!28,opacity=0.62] (1.5291,3.4206) -- (1.7681,3.2347) -- (1.6451,2.9867) -- (1.4061,3.1727) -- cycle;
\draw[cyan!70!black,line width=0.22pt] (1.5291,3.4206) -- (1.7681,3.2347) -- (1.6451,2.9867) -- (1.4061,3.1727) -- cycle;
\fill[cyan!28,opacity=0.62] (1.7697,3.2306) -- (2.0007,3.0508) -- (1.8746,2.8111) -- (1.6435,2.9909) -- cycle;
\draw[cyan!70!black,line width=0.22pt] (1.7697,3.2306) -- (2.0007,3.0508) -- (1.8746,2.8111) -- (1.6435,2.9909) -- cycle;
\fill[cyan!28,opacity=0.62] (2.0021,3.0469) -- (2.2256,2.8730) -- (2.0966,2.6412) -- (1.8731,2.8151) -- cycle;
\draw[cyan!70!black,line width=0.22pt] (2.0021,3.0469) -- (2.2256,2.8730) -- (2.0966,2.6412) -- (1.8731,2.8151) -- cycle;
\fill[cyan!28,opacity=0.62] (2.2269,2.8693) -- (2.4433,2.7010) -- (2.3115,2.4766) -- (2.0952,2.6449) -- cycle;
\draw[cyan!70!black,line width=0.22pt] (2.2269,2.8693) -- (2.4433,2.7010) -- (2.3115,2.4766) -- (2.0952,2.6449) -- cycle;
\fill[cyan!28,opacity=0.62] (2.4444,2.6975) -- (2.6539,2.5345) -- (2.5198,2.3172) -- (2.3103,2.4802) -- cycle;
\draw[cyan!70!black,line width=0.22pt] (2.4444,2.6975) -- (2.6539,2.5345) -- (2.5198,2.3172) -- (2.3103,2.4802) -- cycle;
\fill[cyan!28,opacity=0.62] (1.4078,3.1687) -- (1.6435,2.9909) -- (1.5260,2.7538) -- (1.2903,2.9316) -- cycle;
\draw[cyan!70!black,line width=0.22pt] (1.4078,3.1687) -- (1.6435,2.9909) -- (1.5260,2.7538) -- (1.2903,2.9316) -- cycle;
\fill[cyan!28,opacity=0.62] (1.6451,2.9870) -- (1.8731,2.8151) -- (1.7524,2.5857) -- (1.5244,2.7577) -- cycle;
\draw[cyan!70!black,line width=0.22pt] (1.6451,2.9870) -- (1.8731,2.8151) -- (1.7524,2.5857) -- (1.5244,2.7577) -- cycle;
\fill[cyan!28,opacity=0.62] (1.8745,2.8114) -- (2.0952,2.6449) -- (1.9717,2.4229) -- (1.7509,2.5894) -- cycle;
\draw[cyan!70!black,line width=0.22pt] (1.8745,2.8114) -- (2.0952,2.6449) -- (1.9717,2.4229) -- (1.7509,2.5894) -- cycle;
\fill[cyan!28,opacity=0.62] (2.0965,2.6414) -- (2.3103,2.4802) -- (2.1841,2.2652) -- (1.9703,2.4264) -- cycle;
\draw[cyan!70!black,line width=0.22pt] (2.0965,2.6414) -- (2.3103,2.4802) -- (2.1841,2.2652) -- (1.9703,2.4264) -- cycle;
\fill[cyan!28,opacity=0.62] (2.3115,2.4769) -- (2.5186,2.3206) -- (2.3900,2.1123) -- (2.1829,2.2685) -- cycle;
\draw[cyan!70!black,line width=0.22pt] (2.3115,2.4769) -- (2.5186,2.3206) -- (2.3900,2.1123) -- (2.1829,2.2685) -- cycle;
\fill[cyan!28,opacity=0.62] (1.2919,2.9279) -- (1.5244,2.7577) -- (1.4119,2.5309) -- (1.1794,2.7010) -- cycle;
\draw[cyan!70!black,line width=0.22pt] (1.2919,2.9279) -- (1.5244,2.7577) -- (1.4119,2.5309) -- (1.1794,2.7010) -- cycle;
\fill[cyan!28,opacity=0.62] (1.5259,2.7541) -- (1.7509,2.5894) -- (1.6354,2.3698) -- (1.4103,2.5345) -- cycle;
\draw[cyan!70!black,line width=0.22pt] (1.5259,2.7541) -- (1.7509,2.5894) -- (1.6354,2.3698) -- (1.4103,2.5345) -- cycle;
\fill[cyan!28,opacity=0.62] (1.7523,2.5860) -- (1.9703,2.4264) -- (1.8519,2.2137) -- (1.6339,2.3733) -- cycle;
\draw[cyan!70!black,line width=0.22pt] (1.7523,2.5860) -- (1.9703,2.4264) -- (1.8519,2.2137) -- (1.6339,2.3733) -- cycle;
\fill[cyan!28,opacity=0.62] (1.9716,2.4232) -- (2.1829,2.2685) -- (2.0618,2.0624) -- (1.8506,2.2170) -- cycle;
\draw[cyan!70!black,line width=0.22pt] (1.9716,2.4232) -- (2.1829,2.2685) -- (2.0618,2.0624) -- (1.8506,2.2170) -- cycle;
\fill[cyan!28,opacity=0.62] (2.1840,2.2654) -- (2.3889,2.1155) -- (2.2655,1.9156) -- (2.0606,2.0655) -- cycle;
\draw[cyan!70!black,line width=0.22pt] (2.1840,2.2654) -- (2.3889,2.1155) -- (2.2655,1.9156) -- (2.0606,2.0655) -- cycle;
\fill[cyan!28,opacity=0.62] (1.1811,2.6975) -- (1.4103,2.5345) -- (1.3026,2.3172) -- (1.0733,2.4802) -- cycle;
\draw[cyan!70!black,line width=0.22pt] (1.1811,2.6975) -- (1.4103,2.5345) -- (1.3026,2.3172) -- (1.0733,2.4802) -- cycle;
\fill[cyan!28,opacity=0.62] (1.4118,2.5311) -- (1.6339,2.3733) -- (1.5231,2.1627) -- (1.3010,2.3206) -- cycle;
\draw[cyan!70!black,line width=0.22pt] (1.4118,2.5311) -- (1.6339,2.3733) -- (1.5231,2.1627) -- (1.3010,2.3206) -- cycle;
\fill[cyan!28,opacity=0.62] (1.6353,2.3700) -- (1.8506,2.2170) -- (1.7370,2.0129) -- (1.5217,2.1660) -- cycle;
\draw[cyan!70!black,line width=0.22pt] (1.6353,2.3700) -- (1.8506,2.2170) -- (1.7370,2.0129) -- (1.5217,2.1660) -- cycle;
\fill[cyan!28,opacity=0.62] (1.8518,2.2139) -- (2.0606,2.0655) -- (1.9445,1.8676) -- (1.7357,2.0160) -- cycle;
\draw[cyan!70!black,line width=0.22pt] (1.8518,2.2139) -- (2.0606,2.0655) -- (1.9445,1.8676) -- (1.7357,2.0160) -- cycle;
\fill[cyan!28,opacity=0.62] (2.0618,2.0626) -- (2.2644,1.9186) -- (2.1458,1.7266) -- (1.9433,1.8706) -- cycle;
\draw[cyan!70!black,line width=0.22pt] (2.0618,2.0626) -- (2.2644,1.9186) -- (2.1458,1.7266) -- (1.9433,1.8706) -- cycle;
\fill[cyan!28,opacity=0.62] (1.0749,2.4769) -- (1.3010,2.3206) -- (1.1977,2.1123) -- (0.9716,2.2685) -- cycle;
\draw[cyan!70!black,line width=0.22pt] (1.0749,2.4769) -- (1.3010,2.3206) -- (1.1977,2.1123) -- (0.9716,2.2685) -- cycle;
\fill[cyan!28,opacity=0.62] (1.3025,2.3175) -- (1.5217,2.1660) -- (1.4154,1.9640) -- (1.1962,2.1155) -- cycle;
\draw[cyan!70!black,line width=0.22pt] (1.3025,2.3175) -- (1.5217,2.1660) -- (1.4154,1.9640) -- (1.1962,2.1155) -- cycle;
\fill[cyan!28,opacity=0.62] (1.5231,2.1630) -- (1.7357,2.0160) -- (1.6266,1.8201) -- (1.4140,1.9671) -- cycle;
\draw[cyan!70!black,line width=0.22pt] (1.5231,2.1630) -- (1.7357,2.0160) -- (1.6266,1.8201) -- (1.4140,1.9671) -- cycle;
\fill[cyan!28,opacity=0.62] (1.7369,2.0132) -- (1.9433,1.8706) -- (1.8317,1.6805) -- (1.6253,1.8231) -- cycle;
\draw[cyan!70!black,line width=0.22pt] (1.7369,2.0132) -- (1.9433,1.8706) -- (1.8317,1.6805) -- (1.6253,1.8231) -- cycle;
\fill[cyan!28,opacity=0.62] (1.9444,1.8678) -- (2.1447,1.7294) -- (2.0308,1.5449) -- (1.8305,1.6833) -- cycle;
\draw[cyan!70!black,line width=0.22pt] (1.9444,1.8678) -- (2.1447,1.7294) -- (2.0308,1.5449) -- (1.8305,1.6833) -- cycle;
\fill[blue!25,opacity=0.62] (2.6550,2.5311) -- (2.8579,2.3733) -- (2.7216,2.1627) -- (2.5186,2.3206) -- cycle;
\draw[blue!75!black,line width=0.22pt] (2.6550,2.5311) -- (2.8579,2.3733) -- (2.7216,2.1627) -- (2.5186,2.3206) -- cycle;
\fill[blue!25,opacity=0.62] (2.8589,2.3700) -- (3.0556,2.2170) -- (2.9173,2.0129) -- (2.7205,2.1660) -- cycle;
\draw[blue!75!black,line width=0.22pt] (2.8589,2.3700) -- (3.0556,2.2170) -- (2.9173,2.0129) -- (2.7205,2.1660) -- cycle;
\fill[blue!25,opacity=0.62] (3.0565,2.2139) -- (3.2473,2.0655) -- (3.1071,1.8676) -- (2.9163,2.0160) -- cycle;
\draw[blue!75!black,line width=0.22pt] (3.0565,2.2139) -- (3.2473,2.0655) -- (3.1071,1.8676) -- (2.9163,2.0160) -- cycle;
\fill[blue!25,opacity=0.62] (3.2481,2.0626) -- (3.4332,1.9186) -- (3.2913,1.7266) -- (3.1062,1.8706) -- cycle;
\draw[blue!75!black,line width=0.22pt] (3.2481,2.0626) -- (3.4332,1.9186) -- (3.2913,1.7266) -- (3.1062,1.8706) -- cycle;
\fill[blue!25,opacity=0.62] (3.4339,1.9158) -- (3.6135,1.7760) -- (3.4702,1.5896) -- (3.2905,1.7294) -- cycle;
\draw[blue!75!black,line width=0.22pt] (3.4339,1.9158) -- (3.6135,1.7760) -- (3.4702,1.5896) -- (3.2905,1.7294) -- cycle;
\fill[blue!25,opacity=0.62] (2.5197,2.3175) -- (2.7205,2.1660) -- (2.5897,1.9640) -- (2.3889,2.1155) -- cycle;
\draw[blue!75!black,line width=0.22pt] (2.5197,2.3175) -- (2.7205,2.1660) -- (2.5897,1.9640) -- (2.3889,2.1155) -- cycle;
\fill[blue!25,opacity=0.62] (2.7215,2.1630) -- (2.9163,2.0160) -- (2.7835,1.8201) -- (2.5887,1.9671) -- cycle;
\draw[blue!75!black,line width=0.22pt] (2.7215,2.1630) -- (2.9163,2.0160) -- (2.7835,1.8201) -- (2.5887,1.9671) -- cycle;
\fill[blue!25,opacity=0.62] (2.9172,2.0132) -- (3.1062,1.8706) -- (2.9716,1.6805) -- (2.7825,1.8231) -- cycle;
\draw[blue!75!black,line width=0.22pt] (2.9172,2.0132) -- (3.1062,1.8706) -- (2.9716,1.6805) -- (2.7825,1.8231) -- cycle;
\fill[blue!25,opacity=0.62] (3.1071,1.8678) -- (3.2906,1.7294) -- (3.1542,1.5449) -- (2.9707,1.6833) -- cycle;
\draw[blue!75!black,line width=0.22pt] (3.1071,1.8678) -- (3.2906,1.7294) -- (3.1542,1.5449) -- (2.9707,1.6833) -- cycle;
\fill[blue!25,opacity=0.62] (3.2913,1.7268) -- (3.4695,1.5924) -- (3.3317,1.4131) -- (3.1534,1.5475) -- cycle;
\draw[blue!75!black,line width=0.22pt] (3.2913,1.7268) -- (3.4695,1.5924) -- (3.3317,1.4131) -- (3.1534,1.5475) -- cycle;
\fill[blue!25,opacity=0.62] (2.3900,2.1125) -- (2.5887,1.9671) -- (2.4631,1.7731) -- (2.2644,1.9186) -- cycle;
\draw[blue!75!black,line width=0.22pt] (2.3900,2.1125) -- (2.5887,1.9671) -- (2.4631,1.7731) -- (2.2644,1.9186) -- cycle;
\fill[blue!25,opacity=0.62] (2.5897,1.9642) -- (2.7826,1.8231) -- (2.6549,1.6348) -- (2.4621,1.7760) -- cycle;
\draw[blue!75!black,line width=0.22pt] (2.5897,1.9642) -- (2.7826,1.8231) -- (2.6549,1.6348) -- (2.4621,1.7760) -- cycle;
\fill[blue!25,opacity=0.62] (2.7835,1.8203) -- (2.9707,1.6833) -- (2.8412,1.5005) -- (2.6540,1.6376) -- cycle;
\draw[blue!75!black,line width=0.22pt] (2.7835,1.8203) -- (2.9707,1.6833) -- (2.8412,1.5005) -- (2.6540,1.6376) -- cycle;
\fill[blue!25,opacity=0.62] (2.9715,1.6807) -- (3.1534,1.5475) -- (3.0223,1.3700) -- (2.8404,1.5032) -- cycle;
\draw[blue!75!black,line width=0.22pt] (2.9715,1.6807) -- (3.1534,1.5475) -- (3.0223,1.3700) -- (2.8404,1.5032) -- cycle;
\fill[blue!25,opacity=0.62] (3.1542,1.5451) -- (3.3309,1.4157) -- (3.1983,1.2432) -- (3.0215,1.3725) -- cycle;
\draw[blue!75!black,line width=0.22pt] (3.1542,1.5451) -- (3.3309,1.4157) -- (3.1983,1.2432) -- (3.0215,1.3725) -- cycle;
\fill[blue!25,opacity=0.62] (2.2654,1.9158) -- (2.4621,1.7760) -- (2.3414,1.5896) -- (2.1447,1.7294) -- cycle;
\draw[blue!75!black,line width=0.22pt] (2.2654,1.9158) -- (2.4621,1.7760) -- (2.3414,1.5896) -- (2.1447,1.7294) -- cycle;
\fill[blue!25,opacity=0.62] (2.4630,1.7733) -- (2.6540,1.6376) -- (2.5313,1.4566) -- (2.3403,1.5924) -- cycle;
\draw[blue!75!black,line width=0.22pt] (2.4630,1.7733) -- (2.6540,1.6376) -- (2.5313,1.4566) -- (2.3403,1.5924) -- cycle;
\fill[blue!25,opacity=0.62] (2.6549,1.6350) -- (2.8404,1.5032) -- (2.7158,1.3273) -- (2.5303,1.4592) -- cycle;
\draw[blue!75!black,line width=0.22pt] (2.6549,1.6350) -- (2.8404,1.5032) -- (2.7158,1.3273) -- (2.5303,1.4592) -- cycle;
\fill[blue!25,opacity=0.62] (2.8412,1.5007) -- (3.0215,1.3726) -- (2.8953,1.2017) -- (2.7150,1.3298) -- cycle;
\draw[blue!75!black,line width=0.22pt] (2.8412,1.5007) -- (3.0215,1.3726) -- (2.8953,1.2017) -- (2.7150,1.3298) -- cycle;
\fill[blue!25,opacity=0.62] (3.0223,1.3702) -- (3.1976,1.2456) -- (3.0698,1.0794) -- (2.8945,1.2041) -- cycle;
\draw[blue!75!black,line width=0.22pt] (3.0223,1.3702) -- (3.1976,1.2456) -- (3.0698,1.0794) -- (2.8945,1.2041) -- cycle;
\fill[blue!25,opacity=0.62] (2.1458,1.7268) -- (2.3403,1.5924) -- (2.2242,1.4131) -- (2.0297,1.5475) -- cycle;
\draw[blue!75!black,line width=0.22pt] (2.1458,1.7268) -- (2.3403,1.5924) -- (2.2242,1.4131) -- (2.0297,1.5475) -- cycle;
\fill[blue!25,opacity=0.62] (2.3413,1.5898) -- (2.5303,1.4592) -- (2.4122,1.2851) -- (2.2232,1.4157) -- cycle;
\draw[blue!75!black,line width=0.22pt] (2.3413,1.5898) -- (2.5303,1.4592) -- (2.4122,1.2851) -- (2.2232,1.4157) -- cycle;
\fill[blue!25,opacity=0.62] (2.5312,1.4568) -- (2.7150,1.3298) -- (2.5951,1.1606) -- (2.4113,1.2875) -- cycle;
\draw[blue!75!black,line width=0.22pt] (2.5312,1.4568) -- (2.7150,1.3298) -- (2.5951,1.1606) -- (2.4113,1.2875) -- cycle;
\fill[blue!25,opacity=0.62] (2.7158,1.3275) -- (2.8945,1.2041) -- (2.7729,1.0394) -- (2.5942,1.1629) -- cycle;
\draw[blue!75!black,line width=0.22pt] (2.7158,1.3275) -- (2.8945,1.2041) -- (2.7729,1.0394) -- (2.5942,1.1629) -- cycle;
\fill[blue!25,opacity=0.62] (2.8952,1.2018) -- (3.0690,1.0817) -- (2.9459,0.9216) -- (2.7721,1.0417) -- cycle;
\draw[blue!75!black,line width=0.22pt] (2.8952,1.2018) -- (3.0690,1.0817) -- (2.9459,0.9216) -- (2.7721,1.0417) -- cycle;
\fill[teal!28,opacity=0.62] (2.2308,4.8784) -- (2.4871,4.6417) -- (2.3306,4.3261) -- (2.0743,4.5628) -- cycle;
\draw[teal!70!black,line width=0.22pt] (2.2308,4.8784) -- (2.4871,4.6417) -- (2.3306,4.3261) -- (2.0743,4.5628) -- cycle;
\fill[teal!28,opacity=0.62] (2.4887,4.6359) -- (2.7355,4.4080) -- (2.5756,4.1042) -- (2.3289,4.3321) -- cycle;
\draw[teal!70!black,line width=0.22pt] (2.4887,4.6359) -- (2.7355,4.4080) -- (2.5756,4.1042) -- (2.3289,4.3321) -- cycle;
\fill[teal!28,opacity=0.62] (2.7369,4.4025) -- (2.9746,4.1830) -- (2.8117,3.8904) -- (2.5740,4.1099) -- cycle;
\draw[teal!70!black,line width=0.22pt] (2.7369,4.4025) -- (2.9746,4.1830) -- (2.8117,3.8904) -- (2.5740,4.1099) -- cycle;
\fill[teal!28,opacity=0.62] (2.9759,4.1778) -- (3.2050,3.9662) -- (3.0394,3.6841) -- (2.8103,3.8957) -- cycle;
\draw[teal!70!black,line width=0.22pt] (2.9759,4.1778) -- (3.2050,3.9662) -- (3.0394,3.6841) -- (2.8103,3.8957) -- cycle;
\fill[teal!28,opacity=0.62] (3.2062,3.9612) -- (3.4272,3.7572) -- (3.2592,3.4851) -- (3.0382,3.6891) -- cycle;
\draw[teal!70!black,line width=0.22pt] (3.2062,3.9612) -- (3.4272,3.7572) -- (3.2592,3.4851) -- (3.0382,3.6891) -- cycle;
\fill[teal!28,opacity=0.62] (2.0761,4.5571) -- (2.3289,4.3321) -- (2.1801,4.0321) -- (1.9273,4.2571) -- cycle;
\draw[teal!70!black,line width=0.22pt] (2.0761,4.5571) -- (2.3289,4.3321) -- (2.1801,4.0321) -- (1.9273,4.2571) -- cycle;
\fill[teal!28,opacity=0.62] (2.3305,4.3266) -- (2.5740,4.1099) -- (2.4219,3.8208) -- (2.1784,4.0376) -- cycle;
\draw[teal!70!black,line width=0.22pt] (2.3305,4.3266) -- (2.5740,4.1099) -- (2.4219,3.8208) -- (2.1784,4.0376) -- cycle;
\fill[teal!28,opacity=0.62] (2.5755,4.1047) -- (2.8103,3.8957) -- (2.6552,3.6170) -- (2.4204,3.8260) -- cycle;
\draw[teal!70!black,line width=0.22pt] (2.5755,4.1047) -- (2.8103,3.8957) -- (2.6552,3.6170) -- (2.4204,3.8260) -- cycle;
\fill[teal!28,opacity=0.62] (2.8116,3.8908) -- (3.0382,3.6891) -- (2.8804,3.4202) -- (2.6538,3.6219) -- cycle;
\draw[teal!70!black,line width=0.22pt] (2.8116,3.8908) -- (3.0382,3.6891) -- (2.8804,3.4202) -- (2.6538,3.6219) -- cycle;
\fill[teal!28,opacity=0.62] (3.0394,3.6845) -- (3.2581,3.4898) -- (3.0978,3.2302) -- (2.8791,3.4249) -- cycle;
\draw[teal!70!black,line width=0.22pt] (3.0394,3.6845) -- (3.2581,3.4898) -- (3.0978,3.2302) -- (2.8791,3.4249) -- cycle;
\fill[teal!28,opacity=0.62] (1.9291,4.2517) -- (2.1784,4.0376) -- (2.0368,3.7520) -- (1.7875,3.9662) -- cycle;
\draw[teal!70!black,line width=0.22pt] (1.9291,4.2517) -- (2.1784,4.0376) -- (2.0368,3.7520) -- (1.7875,3.9662) -- cycle;
\fill[teal!28,opacity=0.62] (2.1800,4.0325) -- (2.4204,3.8260) -- (2.2755,3.5506) -- (2.0351,3.7572) -- cycle;
\draw[teal!70!black,line width=0.22pt] (2.1800,4.0325) -- (2.4204,3.8260) -- (2.2755,3.5506) -- (2.0351,3.7572) -- cycle;
\fill[teal!28,opacity=0.62] (2.4218,3.8212) -- (2.6538,3.6219) -- (2.5059,3.3562) -- (2.2739,3.5555) -- cycle;
\draw[teal!70!black,line width=0.22pt] (2.4218,3.8212) -- (2.6538,3.6219) -- (2.5059,3.3562) -- (2.2739,3.5555) -- cycle;
\fill[teal!28,opacity=0.62] (2.6551,3.6174) -- (2.8791,3.4249) -- (2.7285,3.1683) -- (2.5045,3.3608) -- cycle;
\draw[teal!70!black,line width=0.22pt] (2.6551,3.6174) -- (2.8791,3.4249) -- (2.7285,3.1683) -- (2.5045,3.3608) -- cycle;
\fill[teal!28,opacity=0.62] (2.8803,3.4206) -- (3.0967,3.2347) -- (2.9437,2.9867) -- (2.7272,3.1727) -- cycle;
\draw[teal!70!black,line width=0.22pt] (2.8803,3.4206) -- (3.0967,3.2347) -- (2.9437,2.9867) -- (2.7272,3.1727) -- cycle;
\fill[teal!28,opacity=0.62] (1.7893,3.9612) -- (2.0351,3.7572) -- (1.9002,3.4851) -- (1.6544,3.6891) -- cycle;
\draw[teal!70!black,line width=0.22pt] (1.7893,3.9612) -- (2.0351,3.7572) -- (1.9002,3.4851) -- (1.6544,3.6891) -- cycle;
\fill[teal!28,opacity=0.62] (2.0367,3.7525) -- (2.2739,3.5555) -- (2.1357,3.2928) -- (1.8985,3.4898) -- cycle;
\draw[teal!70!black,line width=0.22pt] (2.0367,3.7525) -- (2.2739,3.5555) -- (2.1357,3.2928) -- (1.8985,3.4898) -- cycle;
\fill[teal!28,opacity=0.62] (2.2754,3.5510) -- (2.5045,3.3608) -- (2.3633,3.1071) -- (2.1342,3.2974) -- cycle;
\draw[teal!70!black,line width=0.22pt] (2.2754,3.5510) -- (2.5045,3.3608) -- (2.3633,3.1071) -- (2.1342,3.2974) -- cycle;
\fill[teal!28,opacity=0.62] (2.5058,3.3565) -- (2.7272,3.1727) -- (2.5834,2.9275) -- (2.3619,3.1114) -- cycle;
\draw[teal!70!black,line width=0.22pt] (2.5058,3.3565) -- (2.7272,3.1727) -- (2.5834,2.9275) -- (2.3619,3.1114) -- cycle;
\fill[teal!28,opacity=0.62] (2.7284,3.1687) -- (2.9425,2.9909) -- (2.7962,2.7538) -- (2.5821,2.9316) -- cycle;
\draw[teal!70!black,line width=0.22pt] (2.7284,3.1687) -- (2.9425,2.9909) -- (2.7962,2.7538) -- (2.5821,2.9316) -- cycle;
\fill[teal!28,opacity=0.62] (1.6561,3.6845) -- (1.8985,3.4898) -- (1.7697,3.2302) -- (1.5274,3.4249) -- cycle;
\draw[teal!70!black,line width=0.22pt] (1.6561,3.6845) -- (1.8985,3.4898) -- (1.7697,3.2302) -- (1.5274,3.4249) -- cycle;
\fill[teal!28,opacity=0.62] (1.9001,3.4854) -- (2.1342,3.2974) -- (2.0022,3.0466) -- (1.7681,3.2347) -- cycle;
\draw[teal!70!black,line width=0.22pt] (1.9001,3.4854) -- (2.1342,3.2974) -- (2.0022,3.0466) -- (1.7681,3.2347) -- cycle;
\fill[teal!28,opacity=0.62] (2.1356,3.2932) -- (2.3619,3.1114) -- (2.2270,2.8690) -- (2.0007,3.0508) -- cycle;
\draw[teal!70!black,line width=0.22pt] (2.1356,3.2932) -- (2.3619,3.1114) -- (2.2270,2.8690) -- (2.0007,3.0508) -- cycle;
\fill[teal!28,opacity=0.62] (2.3632,3.1074) -- (2.5821,2.9316) -- (2.4445,2.6972) -- (2.2256,2.8730) -- cycle;
\draw[teal!70!black,line width=0.22pt] (2.3632,3.1074) -- (2.5821,2.9316) -- (2.4445,2.6972) -- (2.2256,2.8730) -- cycle;
\fill[teal!28,opacity=0.62] (2.5833,2.9279) -- (2.7951,2.7577) -- (2.6550,2.5309) -- (2.4433,2.7010) -- cycle;
\draw[teal!70!black,line width=0.22pt] (2.5833,2.9279) -- (2.7951,2.7577) -- (2.6550,2.5309) -- (2.4433,2.7010) -- cycle;
\fill[violet!23,opacity=0.62] (3.4282,3.7525) -- (3.6415,3.5555) -- (3.4714,3.2928) -- (3.2581,3.4898) -- cycle;
\draw[violet!75!black,line width=0.22pt] (3.4282,3.7525) -- (3.6415,3.5555) -- (3.4714,3.2928) -- (3.2581,3.4898) -- cycle;
\fill[violet!23,opacity=0.62] (3.6424,3.5510) -- (3.8485,3.3608) -- (3.6765,3.1071) -- (3.4704,3.2974) -- cycle;
\draw[violet!75!black,line width=0.22pt] (3.6424,3.5510) -- (3.8485,3.3608) -- (3.6765,3.1071) -- (3.4704,3.2974) -- cycle;
\fill[violet!23,opacity=0.62] (3.8493,3.3565) -- (4.0484,3.1727) -- (3.8747,2.9275) -- (3.6756,3.1114) -- cycle;
\draw[violet!75!black,line width=0.22pt] (3.8493,3.3565) -- (4.0484,3.1727) -- (3.8747,2.9275) -- (3.6756,3.1114) -- cycle;
\fill[violet!23,opacity=0.62] (4.0491,3.1687) -- (4.2416,2.9909) -- (4.0664,2.7538) -- (3.8739,2.9316) -- cycle;
\draw[violet!75!black,line width=0.22pt] (4.0491,3.1687) -- (4.2416,2.9909) -- (4.0664,2.7538) -- (3.8739,2.9316) -- cycle;
\fill[violet!23,opacity=0.62] (4.2422,2.9870) -- (4.4284,2.8151) -- (4.2520,2.5857) -- (4.0658,2.7577) -- cycle;
\draw[violet!75!black,line width=0.22pt] (4.2422,2.9870) -- (4.4284,2.8151) -- (4.2520,2.5857) -- (4.0658,2.7577) -- cycle;
\fill[violet!23,opacity=0.62] (3.2591,3.4854) -- (3.4704,3.2974) -- (3.3080,3.0466) -- (3.0967,3.2347) -- cycle;
\draw[violet!75!black,line width=0.22pt] (3.2591,3.4854) -- (3.4704,3.2974) -- (3.3080,3.0466) -- (3.0967,3.2347) -- cycle;
\fill[violet!23,opacity=0.62] (3.4714,3.2932) -- (3.6756,3.1114) -- (3.5112,2.8690) -- (3.3070,3.0508) -- cycle;
\draw[violet!75!black,line width=0.22pt] (3.4714,3.2932) -- (3.6756,3.1114) -- (3.5112,2.8690) -- (3.3070,3.0508) -- cycle;
\fill[violet!23,opacity=0.62] (3.6764,3.1074) -- (3.8739,2.9316) -- (3.7079,2.6972) -- (3.5103,2.8730) -- cycle;
\draw[violet!75!black,line width=0.22pt] (3.6764,3.1074) -- (3.8739,2.9316) -- (3.7079,2.6972) -- (3.5103,2.8730) -- cycle;
\fill[violet!23,opacity=0.62] (3.8747,2.9279) -- (4.0658,2.7577) -- (3.8982,2.5309) -- (3.7071,2.7010) -- cycle;
\draw[violet!75!black,line width=0.22pt] (3.8747,2.9279) -- (4.0658,2.7577) -- (3.8982,2.5309) -- (3.7071,2.7010) -- cycle;
\fill[violet!23,opacity=0.62] (4.0664,2.7541) -- (4.2515,2.5894) -- (4.0825,2.3698) -- (3.8975,2.5345) -- cycle;
\draw[violet!75!black,line width=0.22pt] (4.0664,2.7541) -- (4.2515,2.5894) -- (4.0825,2.3698) -- (3.8975,2.5345) -- cycle;
\fill[violet!23,opacity=0.62] (3.0978,3.2306) -- (3.3070,3.0508) -- (3.1518,2.8111) -- (2.9425,2.9909) -- cycle;
\draw[violet!75!black,line width=0.22pt] (3.0978,3.2306) -- (3.3070,3.0508) -- (3.1518,2.8111) -- (2.9425,2.9909) -- cycle;
\fill[violet!23,opacity=0.62] (3.3080,3.0469) -- (3.5104,2.8730) -- (3.3531,2.6412) -- (3.1508,2.8151) -- cycle;
\draw[violet!75!black,line width=0.22pt] (3.3080,3.0469) -- (3.5104,2.8730) -- (3.3531,2.6412) -- (3.1508,2.8151) -- cycle;
\fill[violet!23,opacity=0.62] (3.5112,2.8693) -- (3.7071,2.7010) -- (3.5481,2.4766) -- (3.3522,2.6449) -- cycle;
\draw[violet!75!black,line width=0.22pt] (3.5112,2.8693) -- (3.7071,2.7010) -- (3.5481,2.4766) -- (3.3522,2.6449) -- cycle;
\fill[violet!23,opacity=0.62] (3.7078,2.6975) -- (3.8975,2.5345) -- (3.7370,2.3172) -- (3.5473,2.4802) -- cycle;
\draw[violet!75!black,line width=0.22pt] (3.7078,2.6975) -- (3.8975,2.5345) -- (3.7370,2.3172) -- (3.5473,2.4802) -- cycle;
\fill[violet!23,opacity=0.62] (3.8982,2.5311) -- (4.0819,2.3733) -- (3.9200,2.1627) -- (3.7363,2.3206) -- cycle;
\draw[violet!75!black,line width=0.22pt] (3.8982,2.5311) -- (4.0819,2.3733) -- (3.9200,2.1627) -- (3.7363,2.3206) -- cycle;
\fill[violet!23,opacity=0.62] (2.9436,2.9870) -- (3.1508,2.8151) -- (3.0022,2.5857) -- (2.7951,2.7577) -- cycle;
\draw[violet!75!black,line width=0.22pt] (2.9436,2.9870) -- (3.1508,2.8151) -- (3.0022,2.5857) -- (2.7951,2.7577) -- cycle;
\fill[violet!23,opacity=0.62] (3.1517,2.8114) -- (3.3522,2.6449) -- (3.2017,2.4229) -- (3.0012,2.5894) -- cycle;
\draw[violet!75!black,line width=0.22pt] (3.1517,2.8114) -- (3.3522,2.6449) -- (3.2017,2.4229) -- (3.0012,2.5894) -- cycle;
\fill[violet!23,opacity=0.62] (3.3531,2.6414) -- (3.5473,2.4802) -- (3.3950,2.2652) -- (3.2008,2.4264) -- cycle;
\draw[violet!75!black,line width=0.22pt] (3.3531,2.6414) -- (3.5473,2.4802) -- (3.3950,2.2652) -- (3.2008,2.4264) -- cycle;
\fill[violet!23,opacity=0.62] (3.5481,2.4769) -- (3.7363,2.3206) -- (3.5823,2.1123) -- (3.3942,2.2685) -- cycle;
\draw[violet!75!black,line width=0.22pt] (3.5481,2.4769) -- (3.7363,2.3206) -- (3.5823,2.1123) -- (3.3942,2.2685) -- cycle;
\fill[violet!23,opacity=0.62] (3.7370,2.3175) -- (3.9194,2.1660) -- (3.7640,1.9640) -- (3.5816,2.1155) -- cycle;
\draw[violet!75!black,line width=0.22pt] (3.7370,2.3175) -- (3.9194,2.1660) -- (3.7640,1.9640) -- (3.5816,2.1155) -- cycle;
\fill[violet!23,opacity=0.62] (2.7962,2.7541) -- (3.0012,2.5894) -- (2.8590,2.3698) -- (2.6539,2.5345) -- cycle;
\draw[violet!75!black,line width=0.22pt] (2.7962,2.7541) -- (3.0012,2.5894) -- (2.8590,2.3698) -- (2.6539,2.5345) -- cycle;
\fill[violet!23,opacity=0.62] (3.0022,2.5860) -- (3.2008,2.4264) -- (3.0565,2.2137) -- (2.8579,2.3733) -- cycle;
\draw[violet!75!black,line width=0.22pt] (3.0022,2.5860) -- (3.2008,2.4264) -- (3.0565,2.2137) -- (2.8579,2.3733) -- cycle;
\fill[violet!23,opacity=0.62] (3.2017,2.4232) -- (3.3942,2.2685) -- (3.2481,2.0624) -- (3.0556,2.2170) -- cycle;
\draw[violet!75!black,line width=0.22pt] (3.2017,2.4232) -- (3.3942,2.2685) -- (3.2481,2.0624) -- (3.0556,2.2170) -- cycle;
\fill[violet!23,opacity=0.62] (3.3950,2.2654) -- (3.5816,2.1155) -- (3.4339,1.9156) -- (3.2473,2.0655) -- cycle;
\draw[violet!75!black,line width=0.22pt] (3.3950,2.2654) -- (3.5816,2.1155) -- (3.4339,1.9156) -- (3.2473,2.0655) -- cycle;
\fill[violet!23,opacity=0.62] (3.5823,2.1125) -- (3.7634,1.9671) -- (3.6142,1.7731) -- (3.4332,1.9186) -- cycle;
\draw[violet!75!black,line width=0.22pt] (3.5823,2.1125) -- (3.7634,1.9671) -- (3.6142,1.7731) -- (3.4332,1.9186) -- cycle;
\fill[blue!95!black] (2.6544,2.5346) circle (2pt) node[above right] {$\ $};
\end{scope}
\begin{scope}[shift={(8.400,0)}]
\node[font=\bfseries] at (2.700,6.480) {$ (x_1,x_2)$-space};
\draw[->] (0,0) -- (5.750,0) node[below] {$x_1$};
\draw[->] (0,0) -- (0,6.150) node[left] {$x_2$};
\fill[orange!4] (0.0000,0.0000) rectangle (3.9706,4.4615);
\draw[orange!70!black,thick] (0.0000,0.0000) rectangle (3.9706,4.4615);
\node[orange!80!black] at (1.3500,0.3569) {$b(p)$};
\draw[orange!95!black,thick] (3.0176,3.7477) rectangle (4.9235,5.1754);
\draw[orange!90!black,line width=0.8pt] (3.9706,3.7477) -- (3.9706,5.1754);
\draw[orange!90!black,line width=0.8pt] (3.0176,4.4615) -- (4.9235,4.4615);
\fill[yellow!45,opacity=0.82] (3.0176,4.4615) rectangle (3.2082,4.6043);
\draw[orange!80!black,line width=0.22pt] (3.0176,4.4615) rectangle (3.2082,4.6043);
\fill[yellow!45,opacity=0.82] (3.2082,4.4615) rectangle (3.3988,4.6043);
\draw[orange!80!black,line width=0.22pt] (3.2082,4.4615) rectangle (3.3988,4.6043);
\fill[yellow!45,opacity=0.82] (3.3988,4.4615) rectangle (3.5894,4.6043);
\draw[orange!80!black,line width=0.22pt] (3.3988,4.4615) rectangle (3.5894,4.6043);
\fill[yellow!45,opacity=0.82] (3.5894,4.4615) rectangle (3.7800,4.6043);
\draw[orange!80!black,line width=0.22pt] (3.5894,4.4615) rectangle (3.7800,4.6043);
\fill[yellow!45,opacity=0.82] (3.7800,4.4615) rectangle (3.9706,4.6043);
\draw[orange!80!black,line width=0.22pt] (3.7800,4.4615) rectangle (3.9706,4.6043);
\fill[yellow!45,opacity=0.82] (3.0176,4.6043) rectangle (3.2082,4.7471);
\draw[orange!80!black,line width=0.22pt] (3.0176,4.6043) rectangle (3.2082,4.7471);
\fill[yellow!45,opacity=0.82] (3.2082,4.6043) rectangle (3.3988,4.7471);
\draw[orange!80!black,line width=0.22pt] (3.2082,4.6043) rectangle (3.3988,4.7471);
\fill[yellow!45,opacity=0.82] (3.3988,4.6043) rectangle (3.5894,4.7471);
\draw[orange!80!black,line width=0.22pt] (3.3988,4.6043) rectangle (3.5894,4.7471);
\fill[yellow!45,opacity=0.82] (3.5894,4.6043) rectangle (3.7800,4.7471);
\draw[orange!80!black,line width=0.22pt] (3.5894,4.6043) rectangle (3.7800,4.7471);
\fill[yellow!45,opacity=0.82] (3.7800,4.6043) rectangle (3.9706,4.7471);
\draw[orange!80!black,line width=0.22pt] (3.7800,4.6043) rectangle (3.9706,4.7471);
\fill[yellow!45,opacity=0.82] (3.0176,4.7471) rectangle (3.2082,4.8898);
\draw[orange!80!black,line width=0.22pt] (3.0176,4.7471) rectangle (3.2082,4.8898);
\fill[yellow!45,opacity=0.82] (3.2082,4.7471) rectangle (3.3988,4.8898);
\draw[orange!80!black,line width=0.22pt] (3.2082,4.7471) rectangle (3.3988,4.8898);
\fill[yellow!45,opacity=0.82] (3.3988,4.7471) rectangle (3.5894,4.8898);
\draw[orange!80!black,line width=0.22pt] (3.3988,4.7471) rectangle (3.5894,4.8898);
\fill[yellow!45,opacity=0.82] (3.5894,4.7471) rectangle (3.7800,4.8898);
\draw[orange!80!black,line width=0.22pt] (3.5894,4.7471) rectangle (3.7800,4.8898);
\fill[yellow!45,opacity=0.82] (3.7800,4.7471) rectangle (3.9706,4.8898);
\draw[orange!80!black,line width=0.22pt] (3.7800,4.7471) rectangle (3.9706,4.8898);
\fill[yellow!45,opacity=0.82] (3.0176,4.8898) rectangle (3.2082,5.0326);
\draw[orange!80!black,line width=0.22pt] (3.0176,4.8898) rectangle (3.2082,5.0326);
\fill[yellow!45,opacity=0.82] (3.2082,4.8898) rectangle (3.3988,5.0326);
\draw[orange!80!black,line width=0.22pt] (3.2082,4.8898) rectangle (3.3988,5.0326);
\fill[yellow!45,opacity=0.82] (3.3988,4.8898) rectangle (3.5894,5.0326);
\draw[orange!80!black,line width=0.22pt] (3.3988,4.8898) rectangle (3.5894,5.0326);
\fill[yellow!45,opacity=0.82] (3.5894,4.8898) rectangle (3.7800,5.0326);
\draw[orange!80!black,line width=0.22pt] (3.5894,4.8898) rectangle (3.7800,5.0326);
\fill[yellow!45,opacity=0.82] (3.7800,4.8898) rectangle (3.9706,5.0326);
\draw[orange!80!black,line width=0.22pt] (3.7800,4.8898) rectangle (3.9706,5.0326);
\fill[yellow!45,opacity=0.82] (3.0176,5.0326) rectangle (3.2082,5.1754);
\draw[orange!80!black,line width=0.22pt] (3.0176,5.0326) rectangle (3.2082,5.1754);
\fill[yellow!45,opacity=0.82] (3.2082,5.0326) rectangle (3.3988,5.1754);
\draw[orange!80!black,line width=0.22pt] (3.2082,5.0326) rectangle (3.3988,5.1754);
\fill[yellow!45,opacity=0.82] (3.3988,5.0326) rectangle (3.5894,5.1754);
\draw[orange!80!black,line width=0.22pt] (3.3988,5.0326) rectangle (3.5894,5.1754);
\fill[yellow!45,opacity=0.82] (3.5894,5.0326) rectangle (3.7800,5.1754);
\draw[orange!80!black,line width=0.22pt] (3.5894,5.0326) rectangle (3.7800,5.1754);
\fill[yellow!45,opacity=0.82] (3.7800,5.0326) rectangle (3.9706,5.1754);
\draw[orange!80!black,line width=0.22pt] (3.7800,5.0326) rectangle (3.9706,5.1754);
\fill[orange!45,opacity=0.82] (3.9706,4.4615) rectangle (4.1612,4.6043);
\draw[orange!90!black,line width=0.22pt] (3.9706,4.4615) rectangle (4.1612,4.6043);
\fill[orange!45,opacity=0.82] (4.1612,4.4615) rectangle (4.3518,4.6043);
\draw[orange!90!black,line width=0.22pt] (4.1612,4.4615) rectangle (4.3518,4.6043);
\fill[orange!45,opacity=0.82] (4.3518,4.4615) rectangle (4.5424,4.6043);
\draw[orange!90!black,line width=0.22pt] (4.3518,4.4615) rectangle (4.5424,4.6043);
\fill[orange!45,opacity=0.82] (4.5424,4.4615) rectangle (4.7329,4.6043);
\draw[orange!90!black,line width=0.22pt] (4.5424,4.4615) rectangle (4.7329,4.6043);
\fill[orange!45,opacity=0.82] (4.7329,4.4615) rectangle (4.9235,4.6043);
\draw[orange!90!black,line width=0.22pt] (4.7329,4.4615) rectangle (4.9235,4.6043);
\fill[orange!45,opacity=0.82] (3.9706,4.6043) rectangle (4.1612,4.7471);
\draw[orange!90!black,line width=0.22pt] (3.9706,4.6043) rectangle (4.1612,4.7471);
\fill[orange!45,opacity=0.82] (4.1612,4.6043) rectangle (4.3518,4.7471);
\draw[orange!90!black,line width=0.22pt] (4.1612,4.6043) rectangle (4.3518,4.7471);
\fill[orange!45,opacity=0.82] (4.3518,4.6043) rectangle (4.5424,4.7471);
\draw[orange!90!black,line width=0.22pt] (4.3518,4.6043) rectangle (4.5424,4.7471);
\fill[orange!45,opacity=0.82] (4.5424,4.6043) rectangle (4.7329,4.7471);
\draw[orange!90!black,line width=0.22pt] (4.5424,4.6043) rectangle (4.7329,4.7471);
\fill[orange!45,opacity=0.82] (4.7329,4.6043) rectangle (4.9235,4.7471);
\draw[orange!90!black,line width=0.22pt] (4.7329,4.6043) rectangle (4.9235,4.7471);
\fill[orange!45,opacity=0.82] (3.9706,4.7471) rectangle (4.1612,4.8898);
\draw[orange!90!black,line width=0.22pt] (3.9706,4.7471) rectangle (4.1612,4.8898);
\fill[orange!45,opacity=0.82] (4.1612,4.7471) rectangle (4.3518,4.8898);
\draw[orange!90!black,line width=0.22pt] (4.1612,4.7471) rectangle (4.3518,4.8898);
\fill[orange!45,opacity=0.82] (4.3518,4.7471) rectangle (4.5424,4.8898);
\draw[orange!90!black,line width=0.22pt] (4.3518,4.7471) rectangle (4.5424,4.8898);
\fill[orange!45,opacity=0.82] (4.5424,4.7471) rectangle (4.7329,4.8898);
\draw[orange!90!black,line width=0.22pt] (4.5424,4.7471) rectangle (4.7329,4.8898);
\fill[orange!45,opacity=0.82] (4.7329,4.7471) rectangle (4.9235,4.8898);
\draw[orange!90!black,line width=0.22pt] (4.7329,4.7471) rectangle (4.9235,4.8898);
\fill[orange!45,opacity=0.82] (3.9706,4.8898) rectangle (4.1612,5.0326);
\draw[orange!90!black,line width=0.22pt] (3.9706,4.8898) rectangle (4.1612,5.0326);
\fill[orange!45,opacity=0.82] (4.1612,4.8898) rectangle (4.3518,5.0326);
\draw[orange!90!black,line width=0.22pt] (4.1612,4.8898) rectangle (4.3518,5.0326);
\fill[orange!45,opacity=0.82] (4.3518,4.8898) rectangle (4.5424,5.0326);
\draw[orange!90!black,line width=0.22pt] (4.3518,4.8898) rectangle (4.5424,5.0326);
\fill[orange!45,opacity=0.82] (4.5424,4.8898) rectangle (4.7329,5.0326);
\draw[orange!90!black,line width=0.22pt] (4.5424,4.8898) rectangle (4.7329,5.0326);
\fill[orange!45,opacity=0.82] (4.7329,4.8898) rectangle (4.9235,5.0326);
\draw[orange!90!black,line width=0.22pt] (4.7329,4.8898) rectangle (4.9235,5.0326);
\fill[orange!45,opacity=0.82] (3.9706,5.0326) rectangle (4.1612,5.1754);
\draw[orange!90!black,line width=0.22pt] (3.9706,5.0326) rectangle (4.1612,5.1754);
\fill[orange!45,opacity=0.82] (4.1612,5.0326) rectangle (4.3518,5.1754);
\draw[orange!90!black,line width=0.22pt] (4.1612,5.0326) rectangle (4.3518,5.1754);
\fill[orange!45,opacity=0.82] (4.3518,5.0326) rectangle (4.5424,5.1754);
\draw[orange!90!black,line width=0.22pt] (4.3518,5.0326) rectangle (4.5424,5.1754);
\fill[orange!45,opacity=0.82] (4.5424,5.0326) rectangle (4.7329,5.1754);
\draw[orange!90!black,line width=0.22pt] (4.5424,5.0326) rectangle (4.7329,5.1754);
\fill[orange!45,opacity=0.82] (4.7329,5.0326) rectangle (4.9235,5.1754);
\draw[orange!90!black,line width=0.22pt] (4.7329,5.0326) rectangle (4.9235,5.1754);
\fill[pink!35,opacity=0.82] (3.0176,3.7477) rectangle (3.2082,3.8905);
\draw[magenta!75!black,line width=0.22pt] (3.0176,3.7477) rectangle (3.2082,3.8905);
\fill[pink!35,opacity=0.82] (3.2082,3.7477) rectangle (3.3988,3.8905);
\draw[magenta!75!black,line width=0.22pt] (3.2082,3.7477) rectangle (3.3988,3.8905);
\fill[pink!35,opacity=0.82] (3.3988,3.7477) rectangle (3.5894,3.8905);
\draw[magenta!75!black,line width=0.22pt] (3.3988,3.7477) rectangle (3.5894,3.8905);
\fill[pink!35,opacity=0.82] (3.5894,3.7477) rectangle (3.7800,3.8905);
\draw[magenta!75!black,line width=0.22pt] (3.5894,3.7477) rectangle (3.7800,3.8905);
\fill[pink!35,opacity=0.82] (3.7800,3.7477) rectangle (3.9706,3.8905);
\draw[magenta!75!black,line width=0.22pt] (3.7800,3.7477) rectangle (3.9706,3.8905);
\fill[pink!35,opacity=0.82] (3.0176,3.8905) rectangle (3.2082,4.0332);
\draw[magenta!75!black,line width=0.22pt] (3.0176,3.8905) rectangle (3.2082,4.0332);
\fill[pink!35,opacity=0.82] (3.2082,3.8905) rectangle (3.3988,4.0332);
\draw[magenta!75!black,line width=0.22pt] (3.2082,3.8905) rectangle (3.3988,4.0332);
\fill[pink!35,opacity=0.82] (3.3988,3.8905) rectangle (3.5894,4.0332);
\draw[magenta!75!black,line width=0.22pt] (3.3988,3.8905) rectangle (3.5894,4.0332);
\fill[pink!35,opacity=0.82] (3.5894,3.8905) rectangle (3.7800,4.0332);
\draw[magenta!75!black,line width=0.22pt] (3.5894,3.8905) rectangle (3.7800,4.0332);
\fill[pink!35,opacity=0.82] (3.7800,3.8905) rectangle (3.9706,4.0332);
\draw[magenta!75!black,line width=0.22pt] (3.7800,3.8905) rectangle (3.9706,4.0332);
\fill[pink!35,opacity=0.82] (3.0176,4.0332) rectangle (3.2082,4.1760);
\draw[magenta!75!black,line width=0.22pt] (3.0176,4.0332) rectangle (3.2082,4.1760);
\fill[pink!35,opacity=0.82] (3.2082,4.0332) rectangle (3.3988,4.1760);
\draw[magenta!75!black,line width=0.22pt] (3.2082,4.0332) rectangle (3.3988,4.1760);
\fill[pink!35,opacity=0.82] (3.3988,4.0332) rectangle (3.5894,4.1760);
\draw[magenta!75!black,line width=0.22pt] (3.3988,4.0332) rectangle (3.5894,4.1760);
\fill[pink!35,opacity=0.82] (3.5894,4.0332) rectangle (3.7800,4.1760);
\draw[magenta!75!black,line width=0.22pt] (3.5894,4.0332) rectangle (3.7800,4.1760);
\fill[pink!35,opacity=0.82] (3.7800,4.0332) rectangle (3.9706,4.1760);
\draw[magenta!75!black,line width=0.22pt] (3.7800,4.0332) rectangle (3.9706,4.1760);
\fill[pink!35,opacity=0.82] (3.0176,4.1760) rectangle (3.2082,4.3188);
\draw[magenta!75!black,line width=0.22pt] (3.0176,4.1760) rectangle (3.2082,4.3188);
\fill[pink!35,opacity=0.82] (3.2082,4.1760) rectangle (3.3988,4.3188);
\draw[magenta!75!black,line width=0.22pt] (3.2082,4.1760) rectangle (3.3988,4.3188);
\fill[pink!35,opacity=0.82] (3.3988,4.1760) rectangle (3.5894,4.3188);
\draw[magenta!75!black,line width=0.22pt] (3.3988,4.1760) rectangle (3.5894,4.3188);
\fill[pink!35,opacity=0.82] (3.5894,4.1760) rectangle (3.7800,4.3188);
\draw[magenta!75!black,line width=0.22pt] (3.5894,4.1760) rectangle (3.7800,4.3188);
\fill[pink!35,opacity=0.82] (3.7800,4.1760) rectangle (3.9706,4.3188);
\draw[magenta!75!black,line width=0.22pt] (3.7800,4.1760) rectangle (3.9706,4.3188);
\fill[pink!35,opacity=0.82] (3.0176,4.3188) rectangle (3.2082,4.4615);
\draw[magenta!75!black,line width=0.22pt] (3.0176,4.3188) rectangle (3.2082,4.4615);
\fill[pink!35,opacity=0.82] (3.2082,4.3188) rectangle (3.3988,4.4615);
\draw[magenta!75!black,line width=0.22pt] (3.2082,4.3188) rectangle (3.3988,4.4615);
\fill[pink!35,opacity=0.82] (3.3988,4.3188) rectangle (3.5894,4.4615);
\draw[magenta!75!black,line width=0.22pt] (3.3988,4.3188) rectangle (3.5894,4.4615);
\fill[pink!35,opacity=0.82] (3.5894,4.3188) rectangle (3.7800,4.4615);
\draw[magenta!75!black,line width=0.22pt] (3.5894,4.3188) rectangle (3.7800,4.4615);
\fill[pink!35,opacity=0.82] (3.7800,4.3188) rectangle (3.9706,4.4615);
\draw[magenta!75!black,line width=0.22pt] (3.7800,4.3188) rectangle (3.9706,4.4615);
\fill[red!28,opacity=0.82] (3.9706,3.7477) rectangle (4.1612,3.8905);
\draw[red!75!black,line width=0.22pt] (3.9706,3.7477) rectangle (4.1612,3.8905);
\fill[red!28,opacity=0.82] (4.1612,3.7477) rectangle (4.3518,3.8905);
\draw[red!75!black,line width=0.22pt] (4.1612,3.7477) rectangle (4.3518,3.8905);
\fill[red!28,opacity=0.82] (4.3518,3.7477) rectangle (4.5424,3.8905);
\draw[red!75!black,line width=0.22pt] (4.3518,3.7477) rectangle (4.5424,3.8905);
\fill[red!28,opacity=0.82] (4.5424,3.7477) rectangle (4.7329,3.8905);
\draw[red!75!black,line width=0.22pt] (4.5424,3.7477) rectangle (4.7329,3.8905);
\fill[red!28,opacity=0.82] (4.7329,3.7477) rectangle (4.9235,3.8905);
\draw[red!75!black,line width=0.22pt] (4.7329,3.7477) rectangle (4.9235,3.8905);
\fill[red!28,opacity=0.82] (3.9706,3.8905) rectangle (4.1612,4.0332);
\draw[red!75!black,line width=0.22pt] (3.9706,3.8905) rectangle (4.1612,4.0332);
\fill[red!28,opacity=0.82] (4.1612,3.8905) rectangle (4.3518,4.0332);
\draw[red!75!black,line width=0.22pt] (4.1612,3.8905) rectangle (4.3518,4.0332);
\fill[red!28,opacity=0.82] (4.3518,3.8905) rectangle (4.5424,4.0332);
\draw[red!75!black,line width=0.22pt] (4.3518,3.8905) rectangle (4.5424,4.0332);
\fill[red!28,opacity=0.82] (4.5424,3.8905) rectangle (4.7329,4.0332);
\draw[red!75!black,line width=0.22pt] (4.5424,3.8905) rectangle (4.7329,4.0332);
\fill[red!28,opacity=0.82] (4.7329,3.8905) rectangle (4.9235,4.0332);
\draw[red!75!black,line width=0.22pt] (4.7329,3.8905) rectangle (4.9235,4.0332);
\fill[red!28,opacity=0.82] (3.9706,4.0332) rectangle (4.1612,4.1760);
\draw[red!75!black,line width=0.22pt] (3.9706,4.0332) rectangle (4.1612,4.1760);
\fill[red!28,opacity=0.82] (4.1612,4.0332) rectangle (4.3518,4.1760);
\draw[red!75!black,line width=0.22pt] (4.1612,4.0332) rectangle (4.3518,4.1760);
\fill[red!28,opacity=0.82] (4.3518,4.0332) rectangle (4.5424,4.1760);
\draw[red!75!black,line width=0.22pt] (4.3518,4.0332) rectangle (4.5424,4.1760);
\fill[red!28,opacity=0.82] (4.5424,4.0332) rectangle (4.7329,4.1760);
\draw[red!75!black,line width=0.22pt] (4.5424,4.0332) rectangle (4.7329,4.1760);
\fill[red!28,opacity=0.82] (4.7329,4.0332) rectangle (4.9235,4.1760);
\draw[red!75!black,line width=0.22pt] (4.7329,4.0332) rectangle (4.9235,4.1760);
\fill[red!28,opacity=0.82] (3.9706,4.1760) rectangle (4.1612,4.3188);
\draw[red!75!black,line width=0.22pt] (3.9706,4.1760) rectangle (4.1612,4.3188);
\fill[red!28,opacity=0.82] (4.1612,4.1760) rectangle (4.3518,4.3188);
\draw[red!75!black,line width=0.22pt] (4.1612,4.1760) rectangle (4.3518,4.3188);
\fill[red!28,opacity=0.82] (4.3518,4.1760) rectangle (4.5424,4.3188);
\draw[red!75!black,line width=0.22pt] (4.3518,4.1760) rectangle (4.5424,4.3188);
\fill[red!28,opacity=0.82] (4.5424,4.1760) rectangle (4.7329,4.3188);
\draw[red!75!black,line width=0.22pt] (4.5424,4.1760) rectangle (4.7329,4.3188);
\fill[red!28,opacity=0.82] (4.7329,4.1760) rectangle (4.9235,4.3188);
\draw[red!75!black,line width=0.22pt] (4.7329,4.1760) rectangle (4.9235,4.3188);
\fill[red!28,opacity=0.82] (3.9706,4.3188) rectangle (4.1612,4.4615);
\draw[red!75!black,line width=0.22pt] (3.9706,4.3188) rectangle (4.1612,4.4615);
\fill[red!28,opacity=0.82] (4.1612,4.3188) rectangle (4.3518,4.4615);
\draw[red!75!black,line width=0.22pt] (4.1612,4.3188) rectangle (4.3518,4.4615);
\fill[red!28,opacity=0.82] (4.3518,4.3188) rectangle (4.5424,4.4615);
\draw[red!75!black,line width=0.22pt] (4.3518,4.3188) rectangle (4.5424,4.4615);
\fill[red!28,opacity=0.82] (4.5424,4.3188) rectangle (4.7329,4.4615);
\draw[red!75!black,line width=0.22pt] (4.5424,4.3188) rectangle (4.7329,4.4615);
\fill[red!28,opacity=0.82] (4.7329,4.3188) rectangle (4.9235,4.4615);
\draw[red!75!black,line width=0.22pt] (4.7329,4.3188) rectangle (4.9235,4.4615);
\draw[orange!95!black,thick] (3.0176,3.7477) rectangle (4.9235,5.1754);
\draw[orange!90!black,line width=0.8pt] (3.9706,3.7477) -- (3.9706,5.1754);
\draw[orange!90!black,line width=0.8pt] (3.0176,4.4615) -- (4.9235,4.4615);

\draw[densely dashed,orange!65] (3.9706,0) -- (3.9706,4.4615) -- (0,4.4615);
\end{scope}
\draw[->,very thick] (7.450,2.900) -- (6.350,2.900) node[midway,above] {$\Phi^{-1}$};
\end{tikzpicture}
\end{minipage}
\caption{Illustration of the transformation of differential areas via the perspective mapping and Jacobian determinant for $p=(2,1)$: single differential elements (top) and inverse mapping of a finite $5\times5$ partition into differential area elements near $p$ (bottom).}
\label{fig:warmup-local-gallery}
\end{figure}
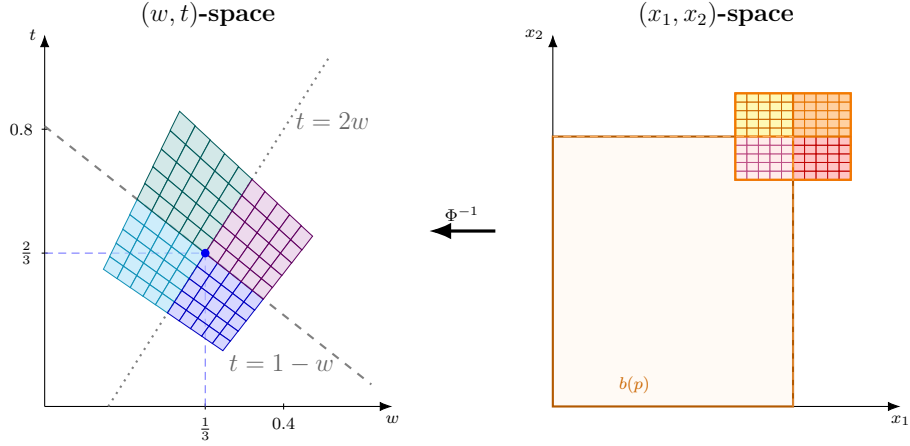

The gallery in Figure~\ref{fig:warmup-local-gallery} explains this determinant geometrically.  In the top row, a small area element is sent, to first order, to a parallelogram whose side vectors are the columns of the Jacobian.  At $q_0=(1/3,2/3)$ the forward area-stretching factor is
\(
  \left|\det D\Phi(q_0)\right|
  =\frac{1}{(2/3)^3}
  =\frac{27}{8}.
\)
Equivalently, at $b(p)=(1/2,1)$ the inverse factor is $(1/2+1)^{-3}=8/27$.  The lower panel repeats the same statement for many small rectangles: if $C$ is centered at $\bar x$, then $\Phi^{-1}(C)\approx\Phi^{-1}(\bar x)+D\Phi^{-1}(\bar x)(C-\bar x)$, so an axis-aligned rectangle becomes a parallelogram with area approximately $(\bar x_1+\bar x_2)^{-3}\operatorname{area}(C)$.  Refining the rectangles gives the change-of-variables formula.
We can now compute the reciprocal-space integral explicitly.  The first part of the box complement contributes
\[
\begin{aligned}
  \int_{1/2}^{\infty}\int_0^{\infty}
  \frac{\mathrm dx_2\,\mathrm dx_1}{(x_1+x_2)^3}
  &=\int_{1/2}^{\infty}
    \left[\frac{1}{2(x_1+x_2)^2}\right]_{x_2=0}^{\infty}
    \mathrm dx_1 \\
  &=\int_{1/2}^{\infty}\frac{1}{2x_1^2}\,\mathrm dx_1
   =1.
\end{aligned}
\]
The second part contributes
\[
\begin{aligned}
  \int_0^{1/2}\int_1^{\infty}
  \frac{\mathrm dx_2\,\mathrm dx_1}{(x_1+x_2)^3}
  &=\int_0^{1/2}\frac{1}{2(x_1+1)^2}\,\mathrm dx_1 
  =\left[-\frac{1}{2(x_1+1)}\right]_{0}^{1/2}
   =\frac16.
\end{aligned}
\]
Thus $R_2(\{p\})=1+\frac16=\frac76$, which agrees with the direct scalarization-side calculation.  This warm-up is the two-dimensional prototype of the higher-dimensional construction: Tchebycheff shadows become complements of anchored reciprocal boxes, while the Jacobian supplies the integration density.

It is an important point that no approximation is used in the value computation and all required integrals are of closed form.  The small rectangles and parallelograms in the gallery only visualize the limiting mechanism behind the change of variables.  Once the determinant is known, the exact integral is evaluated over the reciprocal box complement, which is the representation used in the algorithms below.

\section{Integral \texorpdfstring{$R_2$}{R2} computation in three dimensions}\label{sec:absolute3d}

Next, we turn to the three dimensional case, whereas the dimensions refer to three objective functions in the Pareto optimization problem for which we compute the Integral $R_2$ indicator. We first generalize the perspective mapping to three dimensions (Subsection \ref{sec:pm3d}), then introduce an exact algorithm for the computation of the Integral $R_2$ indicator (Subsection \ref{sec:complexity3d}), and finally discuss how the improvement of the Integral $R_2$ indicator with respect to a reference set can be computed (Subsection \ref{sec:improvement3d}).  
\subsection{Perspective mapping and Jacobian determinant}
\label{sec:pm3d}
For $p\in\RR^3_{>0}$ define the reciprocal box corner
\[
  b(p)=\left(\frac1{p_1},\frac1{p_2},\frac1{p_3}\right),
\]
and for a set $Q\subset\RR^3_{>0}$ define the anchored-box union
\[
  U(Q)=\bigcup_{q\in Q}[0,q_1]\times[0,q_2]\times[0,q_3].
\]
The three-dimensional perspective map is
\[
  x_i=\frac{w_i}{t},\qquad
  w_i=\frac{x_i}{x_1+x_2+x_3},\qquad
  t=\frac1{x_1+x_2+x_3}.
\]

\begin{lemma}[Pointwise box correspondence]\label{lem:pointwise3}
For $p\in\RR^3_{>0}$ and $t>0$, apart from measure-zero boundaries,
\[
  t>g_p(w)
  \quad\Longleftrightarrow\quad
  x\in[0,b_1(p)]\times[0,b_2(p)]\times[0,b_3(p)].
\]
\end{lemma}

\begin{proof}
Using $w_i=x_i/(x_1+x_2+x_3)$ and $t=1/(x_1+x_2+x_3)$, the inequality $t>g_p(w)$ becomes
\[
  \frac1{x_1+x_2+x_3}>
  \max_i \frac{x_ip_i}{x_1+x_2+x_3}.
\]
Multiplying by the positive denominator gives $1>\max_i x_ip_i$, which is equivalent to $x_i<1/p_i$ for all $i$.
\end{proof}

\begin{lemma}[Jacobian determinant]\label{lem:jacobian3}
The perspective transformation has volume element
\[
  \dd w\,\dd t=(x_1+x_2+x_3)^{-4}\dd x_1\dd x_2\dd x_3 .
\]
\end{lemma}

\begin{proof}
Use local coordinates $(w_1,w_2,t)$ with $w_3=1-w_1-w_2$ and inverse
\[
  w_1=\frac{x_1}{s},\qquad w_2=\frac{x_2}{s},\qquad t=\frac1s,
  \qquad s=x_1+x_2+x_3.
\]
The absolute determinant of the Jacobian of $(w_1,w_2,t)$ w.r.t. $(x_1,x_2,x_3)$ is $s^{-4}$.
\end{proof}

\begin{theorem}[Absolute weighted-complement representation]\label{thm:absolute3}
For every nonempty $P\subset\RR^3_{>0}$,
\[
  \Rtwo(P)=
  \int_{\RR^3_{>0}\setminus U(b(P))}
  \frac{\dd x_1\dd x_2\dd x_3}{(x_1+x_2+x_3)^4}.
\]
\end{theorem}

\begin{proof}
By definition,
\[
  \Rtwo(P)=\int_{\DeltaTwo}\tau_P(w)\dd w
  =\int_{\DeltaTwo}\int_0^{\tau_P(w)}\dd t\dd w.
\]
The condition $t<\tau_P(w)$ means $t<g_p(w)$ for every $p\in P$.  By Lemma~\ref{lem:pointwise3}, this is equivalent to saying that $x$ is outside every reciprocal anchored box, i.e. $x\notin U(b(P))$.  Lemma~\ref{lem:jacobian3} supplies the density.
\end{proof}

\subsection{Efficient three-dimensional computation algorithm}\label{sec:complexity3d}
Let
\[
  Q=[\ell_1,u_1]\times[\ell_2,u_2]\times[\ell_3,u_3]
\]
be an axis-aligned box away from the origin; infinite upper bounds are allowed.  An antiderivative of $(x_1+x_2+x_3)^{-4}$ is
\[
  F(x_1,x_2,x_3)=-\frac{1}{6(x_1+x_2+x_3)}.
\]
Therefore
\[
  \mu_3(Q)=\int_Q\frac{\dd x}{(x_1+x_2+x_3)^4}
  =\sum_{\epsilon\in\{0,1\}^3}
  (-1)^{3-(\epsilon_1+\epsilon_2+\epsilon_3)}
  F(x_1^{\epsilon_1},x_2^{\epsilon_2},x_3^{\epsilon_3}),
\]
where $x_i^0=\ell_i$ and $x_i^1=u_i$.  Corner terms with $u_i=\infty$ are interpreted as zero.

Assume now that a hypervolume-style routine returns a disjoint axis-aligned box decomposition of $U(b(P))$ into $M$ boxes in time $T_{\rm decomp}(n,M)$.  Subtracting the single anchor box $U(b(a))$ from one output box creates at most a constant number of boxes in three dimensions.  Each remaining box is integrated in constant time by the formula above.

\begin{theorem}[Output-sensitive bound in three objectives]\label{thm:output-sensitive3}
Given a disjoint $M$-box decomposition of $U(b(P))$, the absolute integral $R_2(P)$ can be computed in
\[
  T_{\rm decomp}(n,M)+O(M)
\]
time and $O(M)$ space after choosing a dominated anchor $a$.
\end{theorem}

\begin{proof}
Compute $\Rtwo(\{a\})$ by the three semi-infinite complement boxes above, which is constant time.  Then compute $I_a(P)$ by Theorem~\ref{thm:weighted-hv3}.  The decomposition costs $T_{\rm decomp}(n,M)$.  For each of the $M$ output boxes, subtracting the anchor box creates only constantly many boxes, and each weighted integral is evaluated by the eight-corner formula in constant time.  Hence the additional work is $O(M)$ and the storage is linear in the emitted decomposition.
\end{proof}

The dimension-sweep algorithm of Fonseca, Paquete, and L{\'o}pez-Ib{\'a}nez~\cite{fonseca2006dimension} computes the hypervolume indicator in asymptotically optimal time $O(n \log n)$ using a dimension-sweep paradigm and a dynamic search tree to maintain points on each level.
In our implementation (see appendix), we use a variant of this algorithm, a tree-free three-dimensional hypervolume algorithm described in a blog-post~\cite{emmerich2025treefree} that computes the integration region by a $z$-sweep and array-only skyline maintenance with path-compressed skip pointers.  The post states $O(n\log n)$ time and $O(n)$ memory for the three-dimensional hypervolume computation and describes the sweep decomposition of the dominated region.  In the case where this routine emits $M=n$ boxes, Theorem~\ref{thm:output-sensitive3} gives
\[
  T(n)=O(n\log n)+O(n)=O(n\log n),
  \qquad
  S(n)=O(n).
\]
By using the hypervolume indicator computation as a box-emitter, we obtain a three-objective complexity bound for computing the absolute integral $R_2$ indicator. As we show in the next section, the algorithm is also asymptotically optimal in time and space complexity.

\subsection{From absolute \texorpdfstring{$R_2$}{R2} to improvement regions}\label{sec:improvement3d}
The absolute representation in Theorem~\ref{thm:absolute3} is conceptually simple, but its integration domain is unbounded.  For computation it is convenient to introduce a dominated anchor $a$ and use
\[
  \Rtwo(P)=\Rtwo(\{a\})-I_a(P),
  \qquad
  I_a(P)=\Rtwo(\{a\})-\Rtwo(P).
\]
The improvement region is bounded away from the singular origin and is exactly a reciprocal anchored-box difference.

\begin{theorem}[Weighted hypervolume representation of improvement]\label{thm:weighted-hv3}
Let $P\subset\RR^3_{>0}$ and let $a\in\RR^3_{>0}$ be componentwise worse than every $p\in P$.  Then
\[
  I_a(P)=\Rtwo(\{a\})-\Rtwo(P)
  =
  \int_{U(b(P))\setminus U(b(a))}
  \frac{\dd x_1\dd x_2\dd x_3}{(x_1+x_2+x_3)^4}.
\]
\end{theorem}

\begin{proof}
Subtract the absolute formula of Theorem~\ref{thm:absolute3} for $P$ from the corresponding formula for the singleton $\{a\}$.  Since $a$ is componentwise worse, $b(a)$ is componentwise smaller and $U(b(a))\subseteq U(b(P))$ whenever $P$ contains at least one point improving upon $a$.  The difference of the two complements is therefore $U(b(P))\setminus U(b(a))$.
\end{proof}

The singleton value $\Rtwo(\{a\})$ itself can be evaluated by the same perspective formula.  If $u=b(a)$, then
\[
\RR^3_{>0}\setminus[0,u_1]\times[0,u_2]\times[0,u_3]
\]
has the disjoint decomposition
\[
[u_1,\infty)\times[0,\infty)\times[0,\infty)
\cup
[0,u_1]\times[u_2,\infty)\times[0,\infty)
\cup
[0,u_1]\times[0,u_2]\times[u_3,\infty).
\]
Thus the absolute integral $R_2(P)$ is obtained by one singleton computation and one bounded weighted box-difference computation.

Note that, unlike the absolute integral $R_2$, the improvement integral $R_2$ is a monotone and submodular set function; thus, it lends itself to submodular greedy optimization with the usual $(1-1/e)$-approximation
ratio. This has been discussed in \cite{emmerich2026threeobjectiveintegralr2subset} alongside with a proof that the subset optimization based on the improvement of the $R_2$ is NP hard already in three dimensions. Note, that instead of a single point it is also straightforward to compute the improvement based integral $R_2$ indicator for a non-empty and non-singleton reference set. See \cite{emmerich2026threeobjectiveintegralr2subset} for details. Finally, the definition generalizes directly to the $N$ dimensional case, as detailed in the next section.


\section{Generalization to \texorpdfstring{$N$}{N} objectives}\label{sec:Ndim}
The perspective transformation is not intrinsically three-dimensional.  The three-objective case is special only because it admits particularly efficient hypervolume box-decomposition algorithms.  The analytic transformation, the weighted box integral, and the output-sensitive reduction all extend to arbitrary dimension $N$.

\subsection{Perspective mapping in $N$ dimensions}
For $p\in\RR^N_{>0}$ define
\[
  b(p)=\left(\frac1{p_1},\ldots,\frac1{p_N}\right),
  \qquad
  \boxop(b(p))=\prod_{i=1}^N[0,1/p_i].
\]
The $N$-dimensional perspective map is
\[
  x_i=\frac{w_i}{t},
  \qquad
  w_i=\frac{x_i}{x_1+\cdots+x_N},
  \qquad
  t=\frac1{x_1+\cdots+x_N}.
\]

\begin{lemma}[N-dimensional perspective correspondence]\label{lem:nd-pointwise}
For $p\in\RR^N_{>0}$, apart from measure-zero boundaries,
\[
  t>g_p(w)
  \quad\Longleftrightarrow\quad
  x\in\boxop(b(p)).
\]
Moreover,
\[
  \dd w\,\dd t=(x_1+\cdots+x_N)^{-(N+1)}\dd x.
\]
\end{lemma}

\begin{proof}
Substituting $w_i=x_i/s$ and $t=1/s$, with $s=x_1+\cdots+x_N$, turns $t>g_p(w)$ into $1>\max_i x_ip_i$, equivalent to $x_i<1/p_i$ for every coordinate.  For the Jacobian, use local coordinates $(w_1,\ldots,w_{N-1},t)$ with $w_N=1-\sum_{i=1}^{N-1}w_i$.  The absolute determinant is $s^{-(N+1)}$.
\end{proof}

\begin{theorem}[Absolute $N$-dimensional representation]\label{thm:nd-absolute}
For every nonempty $P\subset\RR^N_{>0}$,
\[
  \Rtwo(P)=
  \int_{\RR^N_{>0}\setminus U(b(P))}
  \frac{\dd x}{(x_1+\cdots+x_N)^{N+1}}.
\]
\end{theorem}

\begin{proof}
As in the three-dimensional proof,
\[
\Rtwo(P)=\int_{\DeltaN}\int_0^{\tau_P(w)}\dd t\dd w.
\]
The condition $t<\tau_P(w)$ is equivalent to being outside all reciprocal boxes.  Lemma~\ref{lem:nd-pointwise} supplies the density.
\end{proof}

\begin{corollary}[N-dimensional improvement representation]\label{cor:nd-improvement}
If $a\in\RR^N_{>0}$ is componentwise worse than all points in $P$, then
\[
  I_a(P)=\Rtwo(\{a\})-\Rtwo(P)
  =
  \int_{U(b(P))\setminus U(b(a))}
  \frac{\dd x}{(x_1+\cdots+x_N)^{N+1}}.
\]

\end{corollary}

\begin{proposition}[Bidirectional perspective representation and hypervolume reduction]
\label{prop:bidirectional}
Let
\[
  \Phi:\DeltaN\times\RR_{>0}\to\RR^N_{>0},
  \qquad \Phi(w,t)=\frac{w}{t},
\]
where division is componentwise.  Then \(\Phi\) is a smooth bijection, up to the usual boundary sets of measure zero, with inverse
\[
  x\mapsto \left(\frac{x}{x_1+\cdots+x_N},\frac1{x_1+\cdots+x_N}\right).
\]
Let \(Q\subset\RR^N_{>0}\) be a finite set of reciprocal box corners, let \(u\in\RR^N_{>0}\) satisfy \(u_i\le q_i\) for all \(q\in Q\) and all \(i\), and define
\[
  p(q)_i=\frac1{q_i},\qquad a_i=\frac1{u_i},
  \qquad P=\{p(q):q\in Q\}.
\]
Then, apart from boundary sets,
\[
  U(Q)\setminus U(u)
  =
  \Phi\bigl(\{(w,t):w\in\DeltaN,\ \tau_P(w)<t<g_a(w)\}\bigr).
\]
Consequently the ordinary anchored hypervolume of \(Q\) above the lower anchor \(u\) is exactly
\[
\begin{aligned}
  \HV_u(Q)
  &:=\operatorname{vol}(U(Q)\setminus U(u)) \\
  &=\int_{\DeltaN}\int_{\tau_P(w)}^{g_a(w)} t^{-(N+1)}\,\dd t\,\dd w \\
  &=\frac1N\int_{\DeltaN}
  \bigl(\tau_P(w)^{-N}-g_a(w)^{-N}\bigr)\,\dd w .
\end{aligned}
\]
Thus the same perspective coordinates can be used in both directions.  The integral \(R_2\) improvement integrates the interval \(\tau_P(w)<t<g_a(w)\) with unit density in \(t\), whereas ordinary hypervolume integrates the same interval with the inverse-perspective Jacobian density \(t^{-(N+1)}\).  The construction preserves both the number of points and the dimension; for rational positive coordinates it consists only of coordinatewise reciprocals and therefore has linear size and requires \(O(N|Q|)\) arithmetic operations.
\end{proposition}

\begin{proof}
The formula for the inverse follows immediately from \(x_i=w_i/t\), because \(x_1+\cdots+x_N=1/t\).  The pointwise correspondence in Lemma~\ref{lem:nd-pointwise} gives
\[
  x\in U(Q)
  \quad\Longleftrightarrow\quad
  t>\tau_P(w),
\]
and similarly \(x\in U(u)\) is equivalent to \(t>g_a(w)\).  Therefore
\(x\in U(Q)\setminus U(u)\) is equivalent to
\(\tau_P(w)<t<g_a(w)\).  Since \(\dd x=t^{-(N+1)}\dd w\dd t\), integrating the constant density \(1\) over the hypervolume region gives the first integral.  Evaluating the inner integral gives the final expression.  The last statement follows because the transformation \((Q,u)\mapsto(P,a)\) changes no combinatorial structure and applies one reciprocal operation to each coordinate.
\end{proof}

\subsection{Weighted box formula}
For a box $Q=\prod_{i=1}^N[\ell_i,u_i]$ away from the origin, define $x_i^0=\ell_i$ and $x_i^1=u_i$.  The $N$-dimensional weighted box integral is
\[
  \mu_N(Q)
  =
  \int_Q \frac{\dd x}{(x_1+\cdots+x_N)^{N+1}}
  =
  \frac1{N!}
  \sum_{\epsilon\in\{0,1\}^N}
  (-1)^{|\epsilon|}
  \frac1{x_1^{\epsilon_1}+\cdots+x_N^{\epsilon_N}}.
\]
Terms with an infinite coordinate are interpreted as zero.

\begin{proof}
The function
\[
  F_N(x)=\frac{(-1)^N}{N!(x_1+\cdots+x_N)}
\]
has mixed derivative $(x_1+\cdots+x_N)^{-(N+1)}$.  Applying the usual box inclusion-exclusion formula yields the expression above.
\end{proof}

\subsection{Pseudocode using a given box decomposition}
Algorithm~\ref{alg:nd-perspective} assumes that another routine supplies a disjoint box decomposition of the reciprocal dominated region $U(b(P))$.  This routine may be a hypervolume dimension sweep, a recursive Klee-measure routine, a quick hypervolume variant, or any other box-emitting decomposition method.  The perspective algorithm consumes boxes and replaces ordinary volumes by weighted integrals.

\begin{algorithm}[h]
\caption{Perspective evaluation of absolute integral $R_2$ in $N$ objectives}\label{alg:nd-perspective}
\begin{algorithmic}[1]
\Require Dimension $N$, point set $P\subset\RR^N_{>0}$, dominated anchor $a\in\RR^N_{>0}$
\Require A routine \textsc{BoxDecompose} returning disjoint boxes for $U(b(P))$
\Ensure $\Rtwo(P)$
\State $u\gets b(a)$
\State $R_a\gets$ weighted integral of $\RR^N_{>0}\setminus \boxop(u)$ using the $N$ semi-infinite complement boxes
\State $A\gets \boxop(u)$
\State $\mathcal D\gets \textsc{BoxDecompose}(\{b(p):p\in P\})$
\State $I\gets 0$
\ForAll{$Q\in\mathcal D$}
    \State $\mathcal C\gets \textsc{SubtractBox}(Q,A)$
    \ForAll{$C\in\mathcal C$}
        \State $I\gets I+\mu_N(C)$ \Comment{$2^N$-corner weighted integral}
    \EndFor
\EndFor
\State \Return $R_a-I$
\end{algorithmic}
\end{algorithm}

\subsection{Complexity}
\begin{theorem}[N-dimensional output-sensitive complexity]\label{thm:nd-complexity}
Suppose $U(b(P))$ is decomposed into $M$ disjoint axis-aligned boxes in time $T_{\rm decomp}(n,M,N)$.  Then Algorithm~\ref{alg:nd-perspective} computes the absolute integral $R_2(P)$ in
\[
  T_{\rm decomp}(n,M,N)+O(NM+2^N M)
\]
time and $O(M)$ space if the decomposition is stored.  For fixed $N$, this is
\[
  T_{\rm decomp}(n,M,N)+O(M).
\]
\end{theorem}

\begin{proof}
The decomposition step costs $T_{\rm decomp}(n,M,N)$.  Subtracting one anchor box from one emitted box can be done by a slab decomposition with at most $2N$ pieces, so this contributes $O(NM)$ bookkeeping.  Each resulting box is evaluated by summing over its $2^N$ corners.  The $N$ semi-infinite boxes needed for $R_2(\{a\})$ cost $O(N2^N)$, which is dominated by the main term for nontrivial $M$.  Thus the post-processing overhead is $O(NM+2^N M)$.  For fixed dimension all factors depending only on $N$ are constants.
\end{proof}

\begin{remark}[What generalizes]
The perspective mapping generalizes to every dimension $N$.  The special $O(n\log n)$ bound is not a consequence of the analytic formula alone; it is a three-dimensional consequence of having an $O(n\log n)$ box emitter with $M=O(n)$.  In higher dimensions the leading term is inherited from the chosen $N$-dimensional box decomposition, while the $R_2$-specific post-processing is $O(2^N M)$.
\end{remark}

\subsection{Dimension-specific consequences}\label{sec:dimensionconsequences}
The precise fixed-dimension statement is therefore obtained by inserting a known box-decomposition bound into Theorem~\ref{thm:nd-complexity}.  The reciprocal boxes $\boxop(b(p))$ are the same orthogonal objects that occur in search-region and local-upper-bound representations; see Klamroth, Lacour, and Vanderpooten~\cite{klamroth2015searchregion} and the related colored orthogonal range-counting construction of Kaplan, Rubin, Sharir, and Verbin~\cite{kaplan2008colored}.  Suitable high-dimensional hypervolume back ends include the recursive/sweep-type algorithm of While, Hingston, Barone, and Huband~\cite{while2006faster} and the box-decomposition algorithm of Lacour, Klamroth, and Fonseca~\cite{Lacour17}.  For fixed $p=N$ objectives, the nonincremental variant of the latter has time complexity
\[
  O\!\left(n^{\lfloor (N-1)/2\rfloor+1}\right),
\]
and their incremental variant has time complexity
\[
  O\!\left(n^{\lfloor N/2\rfloor+1}\right).
\]
Combining the nonincremental bound with the weighted-integration pass gives the following concise consequences; the $O(2^N M)$ factor is constant per emitted box when $N$ is fixed.

\begin{table}[h]
\centering
\caption{Representative fixed-dimension consequences of the perspective reduction.}\label{tab:dimension-consequences}
\small
\begin{tabular}{p{1.8cm}p{7.3cm}p{4.0cm}}
\toprule
Objectives & Box-decomposition back end & Resulting order for exact integral $R_2$ \\
\midrule
$N=2$ & line sweep / biobjective exact computation & $O(n\log n)$ \\
$N=3$ & three-dimensional dimension sweep with $M=O(n)$ & $O(n\log n)$ \\
$N=4$ & HV4D of Guerreiro, Fonseca, and Emmerich~\cite{Guerreiro12}, or HBDA~\cite{Lacour17} & $O(n^2)$ \\
$N=5,6$ & nonincremental HBDA~\cite{Lacour17} & $O(n^3)$ \\
fixed $N$ & nonincremental HBDA~\cite{Lacour17} & $O\!\left(n^{\lfloor (N-1)/2\rfloor+1}\right)$ \\
\bottomrule
\end{tabular}

\end{table}

The preceding bounds are fixed-dimension upper bounds obtained by constructing a disjoint box decomposition and then evaluating weighted box integrals.  We next record lower bounds for exact value computation itself.  For these statements it is convenient to write
\[
  \Rtwoavg^{(N)}(P)=(N-1)!\,\Rtwo(P)
\]
for the simplex-normalized average over the simplex.  Exact computation of \(\Rtwo\) and exact computation of \(\Rtwoavg^{(N)}\) are equivalent up to a known positive dimension-dependent factor.

\begin{lemma}[Reciprocal-diagonal optimum]\label{lem:reciprocal-optimum}
For \(n\geq1\), among all sets
\[
  Q(S)=\left\{\left(\frac1s,\frac1{1-s}\right):s\in S\right\},
  \qquad S\subset(0,1),\quad |S|\leq n,
\]
the two-dimensional integral value is uniquely minimized by
\[
  S^*=\left\{\frac1{n+1},\frac2{n+1},\ldots,\frac n{n+1}\right\},
\]
and the minimum value is \(1+1/(2n)\).
\end{lemma}

\begin{proof}[Proof sketch]
For an ordered set of parameters, the lower envelope of the Tchebycheff shadows partitions the weight interval into cells of lengths \(\ell_i\).  The best reciprocal-diagonal point assigned to one interval of length \(\ell\) contributes \(\ell+\ell^2/2\).  Hence the total value is at least \(1+\frac12\sum_i\ell_i^2\), which is minimized uniquely when there are \(n\) cells of equal length.  The one-interval minimizer then gives \(s_i=i/(n+1)\).  The full calculation is given in Appendix~\ref{app:value-lower-bound-proofs}.
\end{proof}

\begin{theorem}[Two-dimensional exact value lower bound]\label{thm:value-lower-bound2}
In the algebraic decision-tree model, exact computation of the continuous integral \(R_2\) indicator in two objectives requires \(\Omega(n\log n)\) time in the worst case.  The lower bound holds even for inputs on the reciprocal-diagonal curve.
\end{theorem}

\begin{proof}[Proof sketch]
Reduce the normalized uniform-gap problem to exact \(R_2\) value computation by mapping each input scalar \(s\in(0,1)\) to \((1/s,1/(1-s))\).  By Lemma~\ref{lem:reciprocal-optimum}, the resulting value equals \(1+1/(2n)\) if and only if the input has uniform gaps.  Uniform gap requires \(\Omega(n\log n)\) decisions in the algebraic decision-tree model \cite{beume2009complexity,benor1983lower,preparata1985computational}.  Details are in Appendix~\ref{app:value-lower-bound-proofs}.
\end{proof}

\begin{proposition}[Fixed-dimensional value lower bound by zero padding]\label{prop:fixed-N-padding}
For every fixed \(N\geq2\), exact computation of the normalized continuous integral \(R_2\) value \(\Rtwoavg^{(N)}\) in \(N\) objectives requires \(\Omega(n\log n)\) time in the algebraic decision-tree model, provided nonnegative loss vectors are admitted.
\end{proposition}

\begin{proof}[Proof sketch]
Embed a two-dimensional instance \(P\) as
\[
  \widehat P=\{(p_1,p_2,0,\ldots,0):(p_1,p_2)\in P\}\subset\RR^N_{\geq0}.
\]
For \(N=2\) there is nothing to prove.  For \(N>2\), set \(u=w_1+w_2\), \(\lambda=w_1/u\), and write the remaining coordinates as \((w_3,\ldots,w_N)=(1-u)y\), where \(y\in\Delta_{N-3}\).  The simplex measure transforms as
\[
  \dd w=u(1-u)^{N-3}\,\dd\lambda\,\dd u\,\dd y .
\]
Positive homogeneity contributes another factor \(u\).  Hence the \(u\)-part of the normalized integral is the elementary beta integral
\[
  (N-1)!\operatorname{vol}(\Delta_{N-3})\int_0^1 u^2(1-u)^{N-3}\,\dd u=\frac2N.
\]
Therefore
\[
  \Rtwoavg^{(N)}(\widehat P)=\frac2N\,\Rtwoavg^{(2)}(P).
\]
Thus a faster exact algorithm in fixed dimension \(N\) would imply a faster exact algorithm in two objectives.  See Appendix~\ref{app:value-lower-bound-proofs} for the full integral calculation and the strictly positive-domain caveat.
\end{proof}

\begin{theorem}[Variable-dimensional \texorpdfstring{$\#P$}{sharp-P}-hardness]\label{thm:variable-dim-sharp}
Exact computation of the normalized continuous integral \(R_2\) indicator is \(\#P\)-hard when the number of objectives is part of the input.  The hardness holds under polynomial-time Turing reductions, even for rational strictly positive point coordinates.
\end{theorem}

\begin{proof}[Proof sketch]
We adapt the anchored-box construction of Bringmann and Friedrich~\cite{bringmann2010volume}.  A monotone CNF formula on \(d\) variables is represented by anchored boxes in \([0,1+\tau]^d\), so that uncovered Boolean cells are exactly satisfying assignments.  Reciprocal box corners define a \(d\)-objective integral \(R_2\) instance.  Starting from the scalarization integral over \(\Delta_{d-1}\), we apply the perspective change of variables
\[
  x_i=\frac{w_i}{t},\qquad
  w_i=\frac{x_i}{x_1+\cdots+x_d},\qquad
  t=\frac{1}{x_1+\cdots+x_d}.
\]
The Jacobian gives
\[
  \dd w\,\dd t=(x_1+\cdots+x_d)^{-(d+1)}\,\dd x,
\]
and the subgraph of the lower Tchebycheff envelope becomes the complement of the anchored-box union.  Thus the \(R_2\) value is an ordinary integral with a known perspective weight, not a probabilistic expectation.  Since Boolean cells have Hamming-weight-dependent weighted integrals, a small rational scale parameter \(\tau\) separates the contributions by Hamming weight; exact postprocessing recovers all Hamming-weight counts and hence the number of satisfying assignments.  The full reduction is given in Appendix~\ref{app:hashp-proof}.
\end{proof}

\begin{remark}[Scope of the lower bounds]
Theorem~\ref{thm:value-lower-bound2}, Proposition~\ref{prop:fixed-N-padding}, and Theorem~\ref{thm:variable-dim-sharp} are lower bounds for exact value computation, not only for algorithms that first construct a prescribed decomposition.  They do not rule out faster approximation algorithms, floating-point algorithms under numerical tolerances, or special input models where additional order information is supplied.  Proposition~\ref{prop:fixed-N-padding} uses zero padding and is therefore stated on the closed nonnegative loss domain; the two-dimensional and variable-dimensional constructions use strictly positive coordinates.
\end{remark}

\section{Relation to existing work}
\label{sec:relatedwork}
Shang et al.~\cite{shang2018r2} propose an approximate mapping between a finite-weight $R_2$ indicator and hypervolume. In our mapping, however, the absolute-integral $R_2$ above is not a standard hypervolume: the reciprocal-space density $(x_1+\cdots+x_N)^{-(N+1)}$ is crucial. What we inherit algorithmically from hypervolume is the geometry and box decomposition, while the value-computation lower bounds use two additional ingredients: the uniform-gap lower-bound framework used by Beume et al.~\cite{beume2009complexity} and the algebraic decision-tree theory of Ben-Or~\cite{benor1983lower} and Preparata--Shamos~\cite{preparata1985computational}; in variable dimension, the anchored-box construction of Bringmann and Friedrich~\cite{bringmann2010volume} is adapted to the perspective-weighted density. Conversely, Proposition~\ref{prop:bidirectional} shows that ordinary hypervolume is recovered on the scalarization side by using the inverse-perspective density $t^{-(N+1)}$ over the same Tchebycheff interval. 

In their original work on the biobjective $R_2$ indicator, Sch\"apermeier and Kerschke~\cite{schaepermeier2024r2} emphasise the continuous, integrated variant and Pareto compliance, and propose an efficient $O(n \log n)$ integration algorithm for the biobjective case. Jaszkiewicz and Zielniewicz~\cite{jaszkiewicz2025exact} extend these results (including Pareto compliance) to $N \geq 2$ objectives and introduce Quick R2 (QR2), an exact method that adapts quick hypervolume ideas to $R_2$ computation; QR2 is fast in practice but exponential in $n$ in the worst case. This note presents a complementary approach: use the perspective map to convert the Tchebycheff-shadow integral into weighted integration over a reciprocal anchored-box region. This yields polynomial-time algorithms for fixed dimensions and enables the transfer of many algorithmic results based on unions of anchored boxes from the hypervolume literature. The practical performance of these new algorithms, compared to methods such as QR2, remains to be evaluated.

The 3-D perspective transformation used here was first exploited for the special case of the three-objective integral \(R_2\) indicator in the proof of NP-hardness of the three-objective integral \(R_2\) subset-selection problem~\cite{emmerich2026threeobjectiveintegralr2subset}.  The present note uses the same geometric idea in a different direction and generalizes it to $N\geq 3$ dimensions: instead of a hardness construction for subset selection, it develops an exact computation formula for the integral value by reducing the Tchebycheff-shadow subgraph to weighted integration over a reciprocal box decomposition.The same work proposed a 3-D subdivision algorithm with time complexity $O(n^5)$, which we used as a reference implementation to validate our 3-D implementation.

\section{Conclusion and outlook}\label{sec:conclusion}

This article has shown that the integral $R_2$ indicator can be computed through a perspective transformation of the Tchebycheff-shadow subgraph. Under this transformation, the relevant region becomes the complement, in reciprocal space, of an anchored-box union, and the value $R_2(P)$ is obtained as a weighted measure with density $(x_1+\cdots+x_N)^{-(N+1)}$. A dominated anchor turns the unbounded absolute integral into a single-anchor term minus a bounded weighted hypervolume difference. Thus, the construction yields the absolute integral $R_2$ value itself, and also can be used for computing anchor-based improvements.

In three objectives, this reduction is particularly useful. The fast dimension sweep hypervolume computation that produces a disjoint decomposition of the reciprocal dominated region can serve as a box emitter. Here one can simply replace each box's ordinary volume contribution with its closed-form weighted box integral. Combined with an $O(n\log n)$ three-dimensional box emitter that generates $O(n)$ boxes, this gives an exact $O(n\log n)$-time, $O(n)$-space algorithm for the integral $R_2$ indicator.

The higher-dimensional results have two complementary aspects. Algorithmically, once an $M$-box decomposition is available, the additional weighted-integration overhead is $O(2^N M)$ and is therefore linear in the output size for fixed $N$. The leading term is thus determined by the selected box-decomposition back end. In particular, the available bounds give $O(n^2)$ for $N=4$ using either a four-dimensional hypervolume sweep or nonincremental HBDA \cite{Lacour17}, $O(n^3)$ for $N=5,6$ with nonincremental HBDA, and, more generally, $O(n^{\lfloor (N-1)/2\rfloor+1})$ for fixed $N$ with that back end. On the lower-bound side, exact value computation already has an $\Omega(n\log n)$ lower bound in the algebraic decision-tree model in two objectives, and this lower bound lifts to every fixed $N\geq2$ on the nonnegative loss domain. When $N$ is part of the input, exact normalized integral $R_2$ computation is $\#P$-hard by a perspective-weighted adaptation of the Bringmann--Friedrich anchored-box reduction. These results concern exact value computation; approximation and floating-point computation remain separate questions.

Beyond worst-case complexity, empirical comparisons would also be desirable as a sequel to this work. In particular, the algorithms of Schaepermeier and Kerschke~\cite{schaepermeier2024r2,schaepermeier2025ecj}, the Quick $R_2$ approach of Jaszkiewicz and Zielniewicz~\cite{jaszkiewicz2025exact}, and the new perspective-based box-decomposition algorithms should be compared on representative data sets to assess their practical runtime behavior, memory use, and robustness across dimensions and point-set structures.

\newpage
Our approach also gives also a direct way to compute individual point contributions. Existing contribution algorithms decompose, explicitly or implicitly, the exclusive dominated region of a point into disjoint axis-aligned boxes. In particular, the dimension-sweep low-dimensional contribution algorithm of Emmerich and Fonseca~\cite{emmerich2011contributions} and the up-to-four-dimensional update algorithm of Guerreiro and Fonseca~\cite{guerreiro2017contributions} compute such box decompositions. Consequently, the exclusive boxes emitted by these methods can be reused for individual contributions to compute integral $R_2$ indicator contributions with the perspective mapping approach. 

More generally, any hypervolume-indicator algorithm whose geometric core is a box decomposition can now be transferred to the integral $R_2$ setting by replacing ordinary box volumes with the corresponding weighted box integrals through the perspective transformation. Beyond individual point contributions, this suggests analogous transfers to gradient- and Newton-type methods~\cite{Emmerich2014,Wang2023}: hypervolume-gradient computations are based on lower-dimensional box decompositions of exposed facets, and Newton-type methods build on the same differential structure. Likewise, expected hypervolume improvement methods based on box decompositions~\cite{yang2019expected} can be transported box by box, because the outer expectation is a linear integration operator and the partition cells can be treated separately. A careful numerical implementation, including stable formulas for lower-dimensional and expected-value terms, suitable data structures, and a detailed accounting of constants and degeneracies, remains beyond the scope of this note.

\appendix

\section{Appendix: Proof of value-computation lower bound}\label{app:value-lower-bound-proofs}

This appendix gives the detailed proofs behind Lemma~\ref{lem:reciprocal-optimum}, Theorem~\ref{thm:value-lower-bound2}, and Proposition~\ref{prop:fixed-N-padding}.  Throughout this appendix, the two-objective value is
\[
  \Rtwo(Q)=\int_0^1\min_{q\in Q}\max\{\lambda q_1,(1-\lambda)q_2\}\,d\lambda,
\]
and \(\Rtwoavg^{(N)}\) denotes the simplex-normalized average over the simplex.  Multiplication by the known factor \((N-1)!\) converts between \(\Rtwoavg^{(N)}\) and the unnormalised convention used in the main text.

\subsection{Proof of the reciprocal-diagonal optimum}
Let \(S=\{s_1,\ldots,s_m\}\), where \(m\leq n\), and order the distinct parameters as
\[
  0<s_1<\cdots<s_m<1.
\]
Write
\[
  g_s(\lambda)=\max\left\{\frac{\lambda}{s},\frac{1-\lambda}{1-s}\right\}.
\]
For \(s<t\), the graphs of \(g_s\) and \(g_t\) cross exactly once, at
\[
  \alpha(s,t)=\frac{s}{1-t+s},
\]
and this crossing lies strictly between \(s\) and \(t\).  Moreover, \(g_s(\lambda)<g_t(\lambda)\) for \(\lambda<\alpha(s,t)\) and \(g_t(\lambda)<g_s(\lambda)\) for \(\lambda>\alpha(s,t)\).  Hence the lower envelope of the ordered family \(g_{s_1},\ldots,g_{s_m}\) consists of ordered intervals.  Removing zero-length intervals, we obtain intervals
\[
  I_i=[a_i,b_i],\qquad i=1,\ldots,m,
\]
whose interiors are disjoint, whose union is \([0,1]\) up to endpoints, and whose lengths \(\ell_i=b_i-a_i\) satisfy \(\sum_i\ell_i=1\).

We now solve the local one-interval problem.  Fix \(0\leq a<b\leq1\), write \(\ell=b-a\), and consider
\[
  \min_{0<s<1}\int_a^b g_s(\lambda)\,d\lambda.
\]
The minimizer must lie in \([a,b]\): if \(s<a\), increasing \(s\) decreases \(\lambda/s\) throughout \([a,b]\), and if \(s>b\), decreasing \(s\) decreases \((1-\lambda)/(1-s)\) throughout \([a,b]\).  For \(a\leq s\leq b\), set
\[
  F_{a,b}(s)=
  \int_a^s \frac{1-\lambda}{1-s}\,d\lambda
  +\int_s^b \frac{\lambda}{s}\,d\lambda .
\]
A direct differentiation gives
\[
  F'_{a,b}(s)
  =\frac{\bigl(b-s(1+b-a)\bigr)\bigl(s(a+b-1)-b\bigr)}{2s^2(s-1)^2}.
\]
The second factor has no zero in \([a,b]\), whereas the first factor has the unique zero
\[
  s^*=\frac{b}{1+b-a}=\frac{b}{1+\ell}.
\]
The sign of \(F'_{a,b}\) changes from negative to positive at this point, so \(s^*\) is the unique minimizer.  Substitution gives
\[
  \min_{0<s<1}\int_a^b g_s(\lambda)\,d\lambda
  =F_{a,b}(s^*)
  =\ell+\frac{\ell^2}{2}.
\]
Therefore every set with at most \(n\) distinct reciprocal-diagonal points satisfies
\[
  \Rtwo(Q(S))
  \geq\sum_{i=1}^m\left(\ell_i+\frac{\ell_i^2}{2}\right)
  =1+\frac12\sum_{i=1}^m\ell_i^2.
\]
By Cauchy's inequality,
\[
  \sum_{i=1}^m\ell_i^2\geq\frac1m\geq\frac1n.
\]
Thus \(\Rtwo(Q(S))\geq1+1/(2n)\).
Equality is possible only if \(m=n\), all interval lengths are equal, \(\ell_i=1/n\), and each point \(s_i\) is the unique one-interval minimizer for its envelope interval.  Since the intervals are ordered and cover \([0,1]\), they must be
\[
  I_i=\left[\frac{i-1}{n},\frac{i}{n}\right],
  \qquad i=1,\ldots,n.
\]
For such an interval, the unique one-interval minimizer is
\[
  s_i=\frac{b_i}{1+b_i-a_i}
      =\frac{i/n}{1+1/n}
      =\frac{i}{n+1}.
\]
This proves the stated optimal value and uniqueness.

\subsection{Uniform gap and the two-dimensional lower bound}
The normalized uniform-gap problem takes as input an unordered \(n\)-tuple \((s_1,\ldots,s_n)\in(0,1)^n\).  The answer is yes if the numbers are distinct and, after sorting them as \(0<s_{(1)}<\cdots<s_{(n)}<1\), the augmented sequence with fixed endpoints \(s_{(0)}=0\) and \(s_{(n+1)}=1\) has equal gaps,
\[
  s_{(i)}-s_{(i-1)}=\frac1{n+1},\qquad i=1,\ldots,n+1.
\]
Inputs with repeated values are no-instances.  In the algebraic decision-tree model, this problem requires \(\Omega(n\log n)\) decisions in the worst case.  This is the same ingredient used by Beume et al.~\cite{beume2009complexity} for the two-dimensional hypervolume lower bound; the general algebraic computation-tree framework is due to Ben-Or~\cite{benor1983lower}, and the uniform-gap problem is a standard one-dimensional lower-bound problem discussed by Preparata and Shamos~\cite{preparata1985computational}.

Given an unordered uniform-gap input, construct the reciprocal-diagonal point multiset
\[
  Q(S)=\left\{
  \left(\frac1{s_j},\frac1{1-s_j}\right):j=1,\ldots,n
  \right\}.
\]
This uses only rational operations on the input reals and takes linear time.  Duplicate parameter values give duplicate points and do not change the lower envelope.  By Lemma~\ref{lem:reciprocal-optimum},
\[
  \Rtwo(Q(S))\geq1+\frac1{2n},
\]
with equality if and only if the distinct parameter set is \(\{i/(n+1):i=1,\ldots,n\}\).  This is exactly the yes-condition for uniform gap.  Hence any exact algorithm for two-dimensional integral \(R_2\) value computation running in \(o(n\log n)\) time would decide uniform gap in \(o(n\log n)\) time, a contradiction.

\subsection{Fixed-dimensional zero padding}
Let \(P\subset\RR_{\geq0}^2\) be a two-dimensional instance and construct
\[
  \widehat P
  =\{(p_1,p_2,0,\ldots,0):(p_1,p_2)\in P\}
  \subset\RR_{\geq0}^N .
\]
For \(N=2\), this is the original instance.  Assume therefore that \(N>2\).
For \(w\in\Delta_{N-1}\) and \(\widehat p=(p_1,p_2,0,\ldots,0)\),
\[
  \max_{1\leq j\leq N}w_j\widehat p_j
  =\max\{w_1p_1,w_2p_2\},
\]
because all added coordinates are zero and all coordinates are nonnegative.
Set
\[
  u=w_1+w_2,\qquad \lambda=\frac{w_1}{w_1+w_2}
\]
whenever \(u>0\).  Thus \(w_1=u\lambda\) and \(w_2=u(1-\lambda)\).  Write the remaining coordinates as
\[
  (w_3,\ldots,w_N)=(1-u)y,
  \qquad y\in\Delta_{N-3}.
\]
This decomposes the simplex into the total mass \(u\) assigned to the first two coordinates, the split \(\lambda\) of that mass, and the normalized remaining coordinates \(y\).  In these variables the simplex measure is
\[
  \dd w = u(1-u)^{N-3}\,\dd\lambda\,\dd u\,\dd y .
\]
Here the factor \(u\) comes from splitting the mass \(u\) between \(w_1\) and \(w_2\), and the factor \((1-u)^{N-3}\) comes from scaling the \((N-3)\)-dimensional simplex of the remaining coordinates by \(1-u\).

Positive homogeneity of the weighted Tchebycheff scalarizing value gives
\[
  \min_{\widehat p\in\widehat P}\max_j w_j\widehat p_j
  =u\min_{p\in P}\max\{\lambda p_1,(1-\lambda)p_2\}.
\]
Using the simplex-normalized convention,
\[
  \Rtwoavg^{(N)}(Q)=(N-1)!\int_{\Delta_{N-1}}\min_{q\in Q}\max_j w_jq_j\,\dd w,
\]
we obtain
\[
\begin{aligned}
  \Rtwoavg^{(N)}(\widehat P)
  &=(N-1)!\int_0^1\int_0^1\int_{\Delta_{N-3}}
  u\min_{p\in P}\max\{\lambda p_1,(1-\lambda)p_2\} \\
  &\hspace{4.0cm}\cdot u(1-u)^{N-3}\,\dd y\,\dd\lambda\,\dd u .
\end{aligned}
\]
The integration over \(y\) contributes
\[
  \operatorname{vol}(\Delta_{N-3})=\frac{1}{(N-3)!},
\]
and the \(u\)-integral is the elementary beta integral
\[
  \int_0^1 u^2(1-u)^{N-3}\,\dd u
  =B(3,N-2)
  =\frac{2!(N-3)!}{N!}.
\]
Consequently,
\[
\begin{aligned}
  \Rtwoavg^{(N)}(\widehat P)
  &=(N-1)!\frac{1}{(N-3)!}\frac{2!(N-3)!}{N!}
    \int_0^1\min_{p\in P}\max\{\lambda p_1,(1-\lambda)p_2\}\,\dd\lambda \\
  &=\frac2N\Rtwoavg^{(2)}(P).
\end{aligned}
\]
Equivalently, the factor \(2/N\) is the ratio between the value integral
\[
  B(3,N-2)=\int_0^1u^2(1-u)^{N-3}\,\dd u
\]
and the geometric normalization integral
\[
  B(2,N-2)=\int_0^1u(1-u)^{N-3}\,\dd u,
\]
which is the one-dimensional integral arising from the decomposition by the total mass \(u=w_1+w_2\), before the additional homogeneity factor \(u\) is applied.

A faster exact algorithm for \(\Rtwoavg^{(N)}\) in any fixed dimension \(N>2\) would therefore give a faster exact algorithm in two objectives by padding and multiplying by \(N/2\).  The padding proof uses zero coordinates and is therefore stated for the closed nonnegative loss domain.  If an input model insists on strictly positive loss vectors only, a separate perturbation argument would be needed to remove the zero coordinates while preserving an exact value-reduction statement.

\section{Appendix: Proof of Theorem \ref{thm:variable-dim-sharp} - \texorpdfstring{$\#P$}{sharp-P}-hardness}\label{app:hashp-proof}

We reduce from counting satisfying assignments of a monotone CNF formula, denoted \(\#\mathrm{MON}\text{-}\mathrm{CNF}\).  This is the same counting problem used in the high-dimensional anchored-box volume reduction of Bringmann and Friedrich~\cite{bringmann2010volume}.  Let
\[
  F=\bigwedge_{k=1}^m \bigvee_{i\in C_k} y_i
\]
be a monotone CNF formula on \(d\) variables.  Trivial cases, such as a formula with no clauses or with an empty clause, can be handled directly, so assume that all clauses are nonempty and that \(m\geq1\).

Let \(\tau>0\) be a rational parameter to be chosen below.  For each clause \(C_k\), define an anchored box
\[
  B_k(\tau)=\prod_{i=1}^d [0,q_i^{(k)}(\tau)]
\]
where
\[
  q_i^{(k)}(\tau)=
  \begin{cases}
  1, & i\in C_k,\\
  1+\tau, & i\notin C_k.
  \end{cases}
\]
The cube \([0,1+\tau]^d\) is partitioned, up to boundaries of measure zero, into Boolean cells
\[
  C_z(\tau)=\prod_{i=1}^d I_{z_i}(\tau),
  \qquad z\in\{0,1\}^d,
\]
with \(I_0(\tau)=[0,1]\) and \(I_1(\tau)=[1,1+\tau]\).  A cell \(C_z(\tau)\) is contained in \(B_k(\tau)\) if and only if no variable of the clause \(C_k\) is set to one by \(z\), that is, if and only if clause \(C_k\) is false under \(z\).  Hence the cells of \([0,1+\tau]^d\) not covered by \(\bigcup_kB_k(\tau)\) are exactly the cells corresponding to satisfying assignments of \(F\).

Now form the \(d\)-objective integral \(R_2\) instance
\[
  P_F(\tau)=\left\{
  \left(\frac{1}{q_1^{(k)}(\tau)},\ldots,
        \frac{1}{q_d^{(k)}(\tau)}\right):
  k=1,\ldots,m
  \right\}.
\]
All coordinates are rational and strictly positive.  We now derive the weighted complement integral directly from the scalarization integral.  Let
\[
  \tau_{P_F}(w)=\min_{p\in P_F(\tau)}\max_{1\leq i\leq d}w_i p_i
\]
be the lower weighted Tchebycheff envelope.  Then
\[
  \Rtwoavg^{(d)}(P_F(\tau))
  =(d-1)!\int_{\Delta_{d-1}}\tau_{P_F}(w)\,\dd w
  =(d-1)!\int_{\{(w,t):w\in\Delta_{d-1},\,0<t<\tau_{P_F}(w)\}}\dd w\,\dd t .
\]
Apply the perspective change of variables
\[
  x_i=\frac{w_i}{t},\qquad
  w_i=\frac{x_i}{x_1+\cdots+x_d},\qquad
  t=\frac{1}{x_1+\cdots+x_d}.
\]
The differential transforms as
\[
  \dd w\,\dd t=(x_1+\cdots+x_d)^{-(d+1)}\,\dd x .
\]
For the point corresponding to clause \(C_k\), we have \(p_i=1/q_i^{(k)}(\tau)\).  Hence
\[
  t<\max_i w_i p_i
  \quad\Longleftrightarrow\quad
  \exists i:\; \frac{w_i}{t}>q_i^{(k)}(\tau)
  \quad\Longleftrightarrow\quad
  x\notin B_k(\tau).
\]
Since \(0<t<\tau_{P_F}(w)\) requires this inequality for every clause point, the whole subgraph maps, up to boundary sets of measure zero, to
\[
  \RR_{\geq0}^d\setminus U_F(\tau),
  \qquad U_F(\tau)=\bigcup_{k=1}^m B_k(\tau).
\]
Consequently,
\[
  \Rtwoavg^{(d)}(P_F(\tau))
  =(d-1)!\int_{\RR_{\geq0}^d\setminus U_F(\tau)}
  \frac{\dd x}{(x_1+\cdots+x_d)^{d+1}}.
\]
This derivation uses only the change of variables and its Jacobian.  No probability distribution over weights is introduced.  The factor \((d-1)!\) appears because \(\Rtwoavg^{(d)}\) is simplex-normalized.  The part of this complement outside the cube \([0,1+\tau]^d\) is independent of \(F\).  Denote its normalized contribution by \(K_d(\tau)\).  Equivalently, it is the \(R_2\) value of the single symmetric point \(((1+\tau)^{-1},\ldots,(1+\tau)^{-1})\).  Hence
\[
  K_d(\tau)=\frac{(d-1)!}{1+\tau}
  \int_{\Delta_{d-1}}\max_{1\leq i\leq d} w_i\,\dd w.
\]
This simplex integral has a closed form.  Let \(M(w)=\max_i w_i\).  Using the layer-cake identity for ordinary integrals,
\[
  \int_{\Delta_{d-1}} M(w)\,\dd w
  =\int_0^1 \operatorname{vol}_{d-1}
  \{w\in\Delta_{d-1}:M(w)\geq a\}\,\dd a .
\]
The volume of the clipped simplex \(\{w\in\Delta_{d-1}:w_i\geq a\text{ for all }i\in S\}\) is \((1-|S|a)^{d-1}/(d-1)!\) when \(|S|a\leq1\), and zero otherwise.  Inclusion--exclusion therefore gives the finite expression
\[
  (d-1)!\int_{\Delta_{d-1}}M(w)\,\dd w
  =\int_0^1
  \sum_{r=1}^d(-1)^{r+1}{d\choose r}(1-ra)_+^{d-1}\,\dd a
  =\frac1d\sum_{r=1}^d\frac1r
  =\frac{H_d}{d},
\]
where \(H_d=1+1/2+\cdots+1/d\).  Consequently
\[
  K_d(\tau)=\frac{H_d}{d(1+\tau)}.
\]
This is a known correction term depending only on \(d\) and \(\tau\), not on the formula \(F\).

For \(h=1,\ldots,d\), define the normalized weighted measure of a Boolean cell with Hamming weight \(h\) by
\[
  \gamma_h(\tau)
  =(d-1)!\int_{[1,1+\tau]^h\times[0,1]^{d-h}}
  \frac{dx}{(x_1+\cdots+x_d)^{d+1}}.
\]
For rational \(\tau\), these numbers are exactly computable.  Repeated integration gives a signed sum over the vertices of the box,
\[
  \int_{\prod_i[a_i,b_i]}
  \frac{dx}{(x_1+\cdots+x_d)^{d+1}}
  =\frac{1}{d!}\sum_{\epsilon\in\{0,1\}^d}
  \frac{(-1)^{|\epsilon|}}{\sum_i c_i(\epsilon_i)},
\]
where \(c_i(0)=a_i\) and \(c_i(1)=b_i\), whenever all vertex sums are positive.  This applies here because \(h\geq1\).  The all-zero assignment is not satisfying because all clauses are nonempty.  If
\[
  N_h=\#\{z\models F: |z|=h\},
\]
then the exact oracle value satisfies
\[
  T_F(\tau):=\Rtwoavg^{(d)}(P_F(\tau))-K_d(\tau)
  =\sum_{h=1}^d N_h\gamma_h(\tau).
\]
Thus the \(R_2\) value encodes the numbers of satisfying assignments of each Hamming weight, but with nonuniform weights \(\gamma_h(\tau)\).

It remains to show that these coefficients can be recovered exactly.  Choose
\[
  \tau=\frac{1}{8\,2^d(2d)^{d+1}}.
\]
This rational number has polynomially many bits.  For \(0<\tau\leq1\) and \(h\geq1\), the cell defining \(\gamma_h\) has volume \(\tau^h\), and on this cell the sum \(x_1+\cdots+x_d\) lies between \(1\) and \(2d\).  Consequently,
\[
  \frac{(d-1)!}{(2d)^{d+1}}\tau^h
  \leq \gamma_h(\tau)
  \leq (d-1)!\tau^h .
\]
For any \(r<d\), this implies
\[
  \frac{\sum_{h=r+1}^d N_h\gamma_h(\tau)}{\gamma_r(\tau)}
  \leq
  2^d(2d)^{d+1}\sum_{j\geq1}\tau^j
  <\frac12.
\]
Starting from the exact value \(T_F(\tau)\), recover \(N_1,N_2,\ldots,N_d\) successively.  Suppose \(N_1,\ldots,N_{r-1}\) are already known and set
\[
  R_r=T_F(\tau)-\sum_{h=1}^{r-1}N_h\gamma_h(\tau).
\]
Then
\[
  R_r=N_r\gamma_r(\tau)+\sum_{h=r+1}^dN_h\gamma_h(\tau),
\]
and the tail is strictly smaller than \(\gamma_r(\tau)/2\).  Hence \(N_r\) is the unique integer satisfying
\[
  N_r\leq \frac{R_r}{\gamma_r(\tau)}<N_r+\frac12.
\]
Because all quantities are exact and the functions \(\gamma_h(\tau)\) are computable by the closed-form weighted box integral, this extraction is polynomial-time exact postprocessing.  After recovering all \(N_h\), output \(\#F=\sum_{h=1}^d N_h\).  Thus one exact call to an integral \(R_2\) oracle in dimension \(d\), together with polynomial-time rational preprocessing and postprocessing, solves \(\#\mathrm{MON}\text{-}\mathrm{CNF}\).  Since \(\#\mathrm{MON}\text{-}\mathrm{CNF}\) is \(\#P\)-hard, exact computation of the normalized continuous integral \(R_2\) indicator is \(\#P\)-hard when the number of objectives is part of the input.

\section{Appendix: Implementations}

\subsection{Computation of three dimesional integral R2 indicator}\label{sec:verification}
The accompanying Python script \texttt{integral\_r2\_perspective.py} implements the following functions:
\begin{itemize}
\item \texttt{perspective\_r2\_value(points, anchor)}, which computes the absolute integral $R_2(P)$ by reconstructing it from a dominated anchor;
\item \texttt{perspective\_r2\_improvement(points, anchor)}, which computes $R_2(\{a\})-R_2(P)$ from reciprocal boxes;
\item \texttt{single\_point\_r2\_perspective(anchor)}, which evaluates $R_2(\{a\})$ from the semi-infinite complement decomposition;
\item \texttt{weighted\_box\_integral(box)}, the eight-corner formula in three objectives;
\item \texttt{monte\_carlo\_r2\_value(points)}, a scalarization-side Monte Carlo check;
\item \texttt{decompose\_union\_boxes\_sweep(corners)}, a transparent sweep-based box emitter used for verification.  The weighted summation routine can consume boxes from any hypervolume decomposition routine, including an $O(n\log n)$ tree-free emitter.
\end{itemize}

The verification compares the perspective formula with a subdivision-based exact evaluator over the weight simplex and with Monte Carlo integration.  The subdivision evaluator constructs the arrangement of all affine Tchebycheff pieces on $\Delta_2$ and integrates the active lower envelope cell by cell.

\begin{table}[h]
\centering
\caption{Verification of the absolute $R_2$ perspective implementation. Values are unnormalised simplex integrals.}
\small
\begin{tabular}{lrrrrr}
\toprule
instance & $n$ & perspective $R_2$ & subdivision $R_2$ & abs. diff. & MC half-width\\
\midrule
three\_point\_front & 3 & 0.158359774791 & 0.158359774791 & $2.78\times10^{-17}$ & $2.03\times10^{-4}$\\
four\_point\_front & 4 & 0.141937185975 & 0.141937185975 & $1.11\times10^{-16}$ & $1.72\times10^{-4}$\\
five\_point\_front & 5 & 0.147175934894 & 0.147175934894 & $5.55\times10^{-17}$ & $1.85\times10^{-4}$\\
\bottomrule
\end{tabular}
\end{table}

The exact perspective and subdivision values agree to floating-point precision.  The Monte Carlo estimates fall within the reported sampling uncertainty.  The same CSV file also reports the anchor-normalised improvement values.

\subsection{Numerical checks for the reduction gadgets}\label{app:sanity-checks}

The accompanying archive contains a small Python script, \texttt{sanity\_checks.py}, and a short \texttt{README.md}.  The script is not part of the proof but should serve reviewers for proof verification and help the reader gain intuition on the exact rational examples of how the three gadgets used in the reductions behave as predicted by the theory. The script performs the following checks.
\begin{enumerate}
\item It evaluates the reciprocal-diagonal construction for a small value of \(n\) and checks that the uniformly spaced parameters \(s_i=i/(n+1)\) attain the value \(1+1/(2n)\), while a perturbed instance gives a strictly larger value.
\item It checks the fixed-dimensional padding identity
\[
  R_2^{(N)}(\widehat P)=\frac2N R_2^{(2)}(P)
\]
for one explicit rational two-dimensional instance.
\item It constructs a small monotone CNF formula, builds the Bringmann--Friedrich-style anchored boxes, enumerates the Boolean cells, and verifies that uncovered cells are exactly the satisfying assignments.  It then computes the perspective-weighted cell coefficients \(\gamma_h(\tau)\) exactly and recovers the Hamming-weight counts from the weighted sum by the scale-separation extraction used in Theorem~\ref{thm:variable-dim-sharp}.
\end{enumerate}

A typical run prints output of the following form.
\begin{verbatim}
Reciprocal diagonal check
  n = 5
  uniform R2 = 11/10  expected = 11/10
  perturbed R2 > expected? True

Fixed-dimensional padding check
  N = 5
  R2_N(predicted) = (2/N) R2_2: verified exactly

#P gadget check
  uncovered cells equal satisfying assignments? True
  true Hamming-weight counts:      [0, 1, 3, 1]
  recovered Hamming-weight counts: [0, 1, 3, 1]
\end{verbatim}
The checks use only the Python standard library and exact rational arithmetic via \texttt{fractions.Fraction}.  They are therefore reproducible without numerical quadrature or floating-point tolerances.

\section*{Code availability}
The reference implementation, verification scripts, and verification data are available from the GitHub repository
\href{https://github.com/emmerichmtm/IntegralR2ByPerspectiveMapping}{\texttt{github.com/emmerichmtm/IntegralR2ByPerspectiveMapping}}.
\section*{Acknowledgment}
This research is related to the thematic research area Decision Analytics utilizing Causal Models and Multiobjective Optimization (DEMO, \url{https://www.jyu.fi/demo}) of the University of Jyv\"askyl\"a. The author thanks Bhupinder Saini for a discussion that led to the initial idea on mapping between the integration regions of the Integral R2 Indicator and that of the Hypervolume Indicator.
\section*{Declaration on the Use of Generative AI.}
Generative AI tools were used solely for coding assistance and language refinement. The author assumes full responsibility for all artefacts, datasets, report content, and results.

\end{document}